\documentclass[11pt,journal]{IEEEtran}

\usepackage[margin=0.8in]{geometry}
\usepackage{lipsum}
\usepackage{multirow,setspace}
\usepackage{nicefrac,supertabular}
\usepackage[switch]{lineno}

\onecolumn

\usepackage[boxed,vlined,ruled]{algorithm2e}
\SetAlCapNameFnt{\small}
\SetAlCapFnt{\small}
\usepackage{mlmodern}
\usepackage[T1]{fontenc}
\usepackage{url}
\usepackage{ifthen}
\usepackage{fancybox}

\usepackage[usenames,dvipsnames,table]{xcolor}

\usepackage[scaled=0.9]{beramono}

\usepackage{array,graphicx}
\usepackage{amsfonts,cite,euscript}
\usepackage[cmex10]{amsmath}
\usepackage[amsmath]{empheq}
\usepackage{cases,balance,cite}
\usepackage{booktabs,mathrsfs,bbm,amssymb,soul,lipsum}
\usepackage{amsthm}
\usepackage{dsfont,pifont}

\usepackage{qrcode}

\sethlcolor{black!20}

\usepackage[inline, shortlabels]{enumitem}
\setlist{itemjoin ={,\enspace}, itemjoin* = {, and\enspace}}

\graphicspath{{./Figures/}}

\usepackage[percent]{overpic}
\newcommand{\bpara}[1]		{\medskip \noindent {\bf #1}}
\newcommand{\bparab}[1] 	{\noindent {\bf #1}}
\newcommand{\spara}[1]		{\smallskip \noindent {\bf{#1}}}
\newcommand{\ipara}[1]		{\medskip \noindent {---\emph{#1}}}

\definecolor{modulo}{RGB}{72,0,255}
\definecolor{tof}{RGB}{255,0,63}

\usepackage{hyperref}
\hypersetup{
    colorlinks=true,
    linkcolor=black,
    filecolor=black,      
    urlcolor=black,
    citecolor = blue,
}

\theoremstyle{plain}
\newtheorem*{theorem*}{Theorem}
\newtheorem{theorem}{Theorem}
\newtheorem{definition}{Definition}
\newtheorem{lemma}{Lemma}
\newtheorem{cor}{Corollary}

\newtheorem*{eg}{Example}

\usepackage{physics}
\usepackage{tikz}
\usepackage{mathdots}
\usepackage{yhmath}
\usepackage{cancel}
\usepackage{multirow}

\usepackage{tabularx}
\usepackage{extarrows}
\usepackage{booktabs}
\usetikzlibrary{fadings}
\usetikzlibrary{patterns}
\usetikzlibrary{shadows.blur}
\usetikzlibrary{shapes}

\def\N					{\mathbb N}
\def\Z					{\mathbb Z}
\def\R					{\mathbb R}

\def\iR					{\in \mathbb R}

\def\e					{{e}}
\def\ind					{\mathbbmtt{1}}
\def\DE					{\stackrel{\rm{def}}{=}}

\def\ts{T_\mathsf{S}}

\def\cndp{P}
\def\cndpAM{P_{\mathsf{AM}}}
\def\blb{\Omega_{\mathsf{U},\mathsf{L}}}
\def\os{\Omega_{\mathsf{S}}}
\def\ol{\Omega_{\mathsf{L}}}
\def\ou{\Omega_{\mathsf{U}}}

\def\oup{\Omega_{\mathsf{U}}^{g}}

\def\ind{{{\pmb{1}}}}
\def\tbp{\mathsf{BP}}

\def\ql{Q_{\mathsf{L}}}
\def\qu{Q_{\mathsf{U}}}
\def\qlp{Q_{\mathsf{L}}^{g}}
\def\qup{Q_{\mathsf{U}}^{g}}

\def\qlb{Q_{\mathsf{U},\mathsf{L}}}

\def\TUS					{T_{\mathsf{US}}}
\def\TNQ					{T_{\mathsf{NS}}}
\def\TFD					{T_{\mathsf{FD}}}

\usepackage{xspace}
\def\adc					{ADC\xspace}
\def\adcs					{ADCs\xspace}

\def\usf					{USF\xspace}
\def\usalg					{{{\selectfont\texttt{US}--\texttt{Alg}}}\xspace}
\def\madcs					{{$\mathscr{M}_\lambda$--{\fontsize{11pt}{11pt}\selectfont\texttt{ADCs}}}\xspace}
\def\madc					{{$\mathscr{M}$\hspace{-0.15em}--{\selectfont\texttt{ADC}}}\xspace}
\def\usfp					{{\fontsize{11pt}{11pt}\selectfont\texttt{US}-\texttt{FP}}\xspace}
\def\usadc{\madc}

\def\l						{\left(}
\def\r						{\right)}

\def\ind					{\mathds{1}}

\def\etal					{\emph{et al}.~}
\def\eg					{\emph{e.g.~}}
\def\ie					{\emph{i.e.}}

\def\ddvmat				{\bs{\EuScript{D}}_{\hat{\bs{\bar{\varphi}}}}}

\newcommand\bs[1]			{\boldsymbol{#1}}

\newcommand\eset[3]		{\mathbb{E}_{\sqb{{#1},{#2}},{#3}}}
\newcommand\mcal[1]		{\mathcal{#1}}

\newcommand{\LL}[1] {L^{#1}\rob{\R}} 

\newcommand{\ft}[2] {\widehat{#1}\rob{#2}} 
\newcommand{\dft}[2] {\widehat{#1}\sqb{#2}} 
\newcommand{\ftsub}[3] {\widehat{#1}_{\mathsf{#2}}\rob{#3}} 
\newcommand{\dftsub}[3] {\widehat{#1}_{\mathsf{#2}}\sqb{#3}} 

\newcommand{\sqb}[1]    {\left[ #1 \right]} 
\newcommand{\cb}[1]{\left\lbrace #1 \right\rbrace} 
\newcommand\rob[1]      {\l #1 \r} 
\newcommand{\vrb}[1]{\left\lvert #1 \right\rvert}
\newcommand{\PW}[1]{\mathsf{PW}_{#1}}

\newcommand{\mat}[1]    {\mathbf{#1}}
\newcommand\fig[1]			{Fig.~\ref{#1}}
\newcommand\rftab[1]			{Table~\ref{#1}}

\newcommand{\fes}[1]			{\left[\kern-0.15em\left[#1\right]\kern-0.15em\right]}
\newcommand{\fe}[1]		{\left[\kern-0.30em\left[#1\right]\kern-0.30em\right]}
\newcommand{\flr}[1]		{\left\lfloor #1 \right\rfloor}
\newcommand{\cil}[1]       {\left \lceil #1 \right \rceil}

\newcommand{\nrm}[2]   {{||#1||}_{#2}} 
\newcommand{\maxn}[1]   {{||#1||}_\infty} 

\newcommand{\MO}[1]		{\mathscr{M}_\lambda ({#1} )}
\newcommand{\MONI}[1]		{{\widetilde{\mathscr{M}}}_\lambda ({#1} )}
\newcommand{\MOh}[0]{\mathscr{M}_{\boldsymbol{\mathsf{H}}}}

\newcommand{\VO}[1]		{\varepsilon_{#1}}

\newcommand{\BL}[1]		{#1 \in \mathcal{B}_{\Omega}}
\newcommand{\BP}[1]		{#1\in\mathcal{B}_{\left(\ol, \ou \right)}}

\newcommand{\EQc}[1]		{\stackrel{(\ref{#1})}{=}}

\newcommand{\mse}[2]    {\EuScript{E}\rob{{#1,#2}}}

\newcommand{\equalref}[1] {\stackrel{\textsf{#1}}{=}}

\renewcommand\bar\underline
\renewcommand\hat\widehat
\renewcommand\geq\geqslant
\renewcommand\leq\leqslant
\renewcommand\Psi\ddvmat

\renewcommand\tilde\widetilde

\makeatletter
\def\moverlay{\mathpalette\mov@rlay}
\def\mov@rlay#1#2{\leavevmode\vtop{%
   \baselineskip\z@skip \lineskiplimit-\maxdimen
   \ialign{\hfil$\m@th#1##$\hfil\cr#2\crcr}}}
\newcommand{\charfusion}[3][\mathord]{
    #1{\ifx#1\mathop\vphantom{#2}\fi
        \mathpalette\mov@rlay{#2\cr#3}
      }
    \ifx#1\mathop\expandafter\displaylimits\fi}
\makeatother

\newcommand{\cupdot}{\charfusion[\mathbin]{\cup}{\cdot}}

\usepackage[most]{tcolorbox}
\newtcbtheorem{stp}{Arc.}%
{ 
    grow to left by=-0.2cm,grow to right by=0.45cm,
    bottom = -0.25cm,
    enhanced,frame empty,interior empty,
    colframe=black!60,
    borderline west={1pt}{1pt}{black},
    left=0.1cm,
    attach boxed title to top left={yshift=-2mm,xshift=-2mm},
    coltitle=black,
    fonttitle={\bf \small},
    colbacktitle=Goldenrod!15,
    boxed title style={boxrule=.5pt,sharp corners,size=small}}{Steps}

\newtcbtheorem{stp2}{Arc.}%
{ 
    grow to left by=-0.2cm,grow to right by=0.45cm,
    bottom = -0.25cm,
    enhanced,frame empty,interior empty,
    colframe=black!60,
    borderline west={1pt}{1pt}{black},
    left=0.1cm,
    attach boxed title to top left={yshift=-2mm,xshift=-2mm},
    coltitle=black,
    fonttitle={\bf \small},
    colbacktitle=Goldenrod!15,
    boxed title style={boxrule=.5pt,sharp corners,size=small}}{Steps}

\ifCLASSINFOpdf
\else
\fi
\begin{document}

\title{Unlimited Sampling of Bandpass Signals: \\
Computational Demodulation via Undersampling}

\author{Gal Shtendel, Dorian Florescu, and
        Ayush Bhandari

\thanks{This work is supported by the UK Research and Innovation council's \emph{Future Leaders Fellowship} program ``Sensing Beyond Barriers'' (MRC Fellowship award no.~MR/S034897/1). Project page for (future) release of hardware design, code and data: \href{https://bit.ly/USF-Link}{\texttt{https://bit.ly/USF-Link}}.}
\thanks{The authors are with the Dept. of Electrical and Electronic Engineering, Imperial College London, South Kensington, London SW7 2AZ, UK. (Email: \texttt{\{g.shtendel21,d.florescu,a.bhandari\}@imperial.ac.uk} or \texttt{ayush@alum.mit.edu}).}

%
\thanks{Manuscript Submitted: 20XX.}}

\markboth{Journal of \LaTeX\ Class Files,~Vol.~XX, No.~X, Month~20XX}%
{Manuscript: Unlimited Sensing Radar}
%



\maketitle

\tableofcontents

\newpage

\begin{abstract}
Bandpass signals are an important sub-class of bandlimited signals that naturally arise in a number of application areas but their high-frequency content poses an acquisition challenge. Consequently, ``Bandpass Sampling Theory'' has been investigated and applied in the literature. 
In this paper, we consider the problem of modulo sampling of bandpass signals with the main goal of sampling and recovery of high dynamic range inputs. Our work is inspired by the  Unlimited Sensing Framework (USF). In the USF, the modulo operation folds high dynamic range inputs into low dynamic range, modulo samples. This fundamentally avoids signal clipping. Given that the output of the modulo nonlinearity is non-bandlimited, bandpass sampling conditions never hold true. Yet, we show that bandpass signals can be recovered from a modulo representation despite the inevitable aliasing. Our main contribution includes proof of sampling theorems for recovery of bandpass signals from an undersampled representation, reaching sub-Nyquist sampling rates. On the recovery front, by considering both time- and frequency-domain perspectives, we provide a holistic view of the modulo bandpass sampling problem. On the hardware front, we include ideal, non-ideal and generalized modulo folding architectures that arise in the hardware implementation of modulo analog-to-digital converters. Numerical simulations corroborate our theoretical results. Bridging the theory--practice gap, we validate our results using hardware experiments, thus demonstrating the practical effectiveness of our methods. 
\end{abstract}

\begin{IEEEkeywords}
Analog-to-digital conversion (ADC), approximation, bandpass sampling, modulo, Shannon sampling theory.
\end{IEEEkeywords}

%

\IEEEpeerreviewmaketitle

\newpage

\onehalfspacing
\section{Introduction}

\IEEEPARstart{D}{igital} signal acquisition methodology is the workhorse of all modern world electronic systems. This technology is pivoted on 
Shannon's Sampling Framework wherein point-wise samples are digitized using the so-called analog-to-digital converter or the \adc. For the theory and practice to work in perfect tandem, it is imperative that the input signal's dynamic range is matched to the \adc's dynamic range. When this is not the case, even in noiseless conditions, the resulting samples are distorted due to saturation or clipping. Such a scenario leading to permanent loss of information poses a fundamental bottleneck in all digital systems.

\begin{figure*}[!t]
    \centering
    \includegraphics[width =1\textwidth]{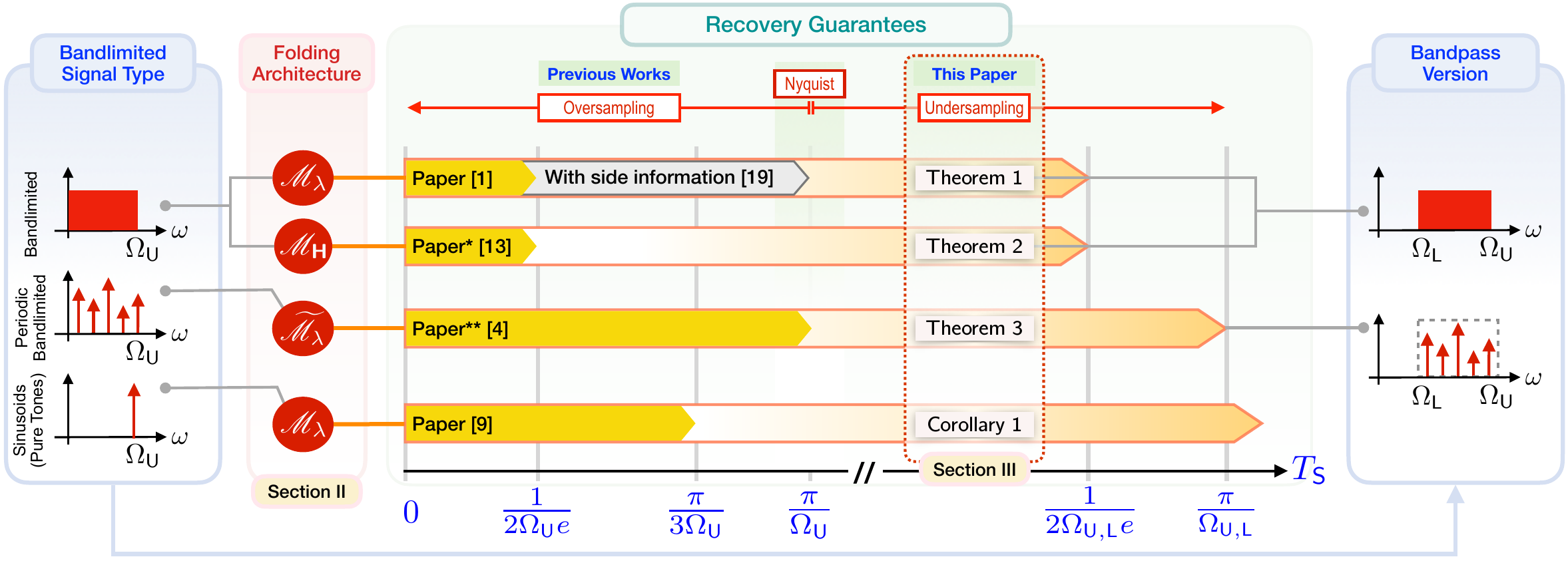}
    \caption{Summary of recovery guarantees for commonly occurring types of bandlimited signals in the context of unlimited sensing framework (\usf) together with different folding architectures denoted by $\mathscr{M}_{\lambda},\tilde{\mathscr{M}}_{\lambda}, \MOh$. We also summarize our contributions in terms of extending the results for bandlimited signal type to their corresponding bandpass case, especially in the undersampled regime.}
    \label{fig:LS}
\end{figure*}

To overcome this bottleneck, the {\bf Unlimited Sensing Framework} (\usf) \cite{Bhandari:2017:C,Bhandari:2020:Pat,Bhandari:2020:Ja,Bhandari:2021:J,Bhandari:2022:J,Florescu:2022:Ja} has been recently proposed as an alternative digital acquisition technology. Unlike previous sampling approaches where \emph{hardware} and \emph{reconstruction algorithms} were mostly decoupled, the \usf is pivoted on the theme of \emph{computational sensing} \cite{Bhandari:2022:Book} as it harnesses a co-design of hardware and algorithms. In particular, 
\begin{enumerate}[leftmargin = 30pt, label = $\bullet$]
  \item {\bf Encoding.} Signal folding\footnote{To get a sense of the hardware functionality of the modulo \adcs built in \cite{Bhandari:2021:J}, we refer the reader to the live YouTube demonstration for (a) Sinusoidal Waveform 
  \href{https://youtu.be/JuZg80gUr8M}{\texttt{https://youtu.be/JuZg80gUr8M}} 
  and (b) Arbitrary Bandlimited Waveform 
  \href{https://youtu.be/prV40WlzHh4}{\texttt{https://youtu.be/prV40WlzHh4}} .} 
  is implemented in the hardware pipeline so that any arbitrary high-dynamic-range (HDR), continuous-time input is folded or wrapped within the dynamic range of the \adc. This results in low-dynamic-range samples and entails that signal clipping or saturation problem is eliminated in the encoding stage.

\item {\bf Decoding.} Given folded samples, the input HDR signal is mathematically \emph{unfolded} using recovery algorithms. To this end, one of the first recovery guarantees was provided in \cite{Bhandari:2017:C,Bhandari:2020:Ja}. It was shown that bandlimited signals can be recovered from a constant factor oversampling of modulo samples with a sampling rate that is independent of the threshold $\lambda$. We refer to \fig{fig:LS} for a summary of known recovery rates for different bandlimited signal types \cite{Bhandari:2017:C,Bhandari:2018:C,Bhandari:2020:Ja,Bhandari:2021:J,Shtendel:2022:C}. 
\end{enumerate}

To validate the practical utility of the \usf, we have implemented signal folding in hardware using different non-linearities including 
\begin{enumerate}[leftmargin =30pt,label=\ding{224}]
\item Modulo \adcs \cite{Bhandari:2021:J,Bhandari:2022:J,Bhandari:2022:C,Beckmann:2022:J,Liu:2022},
\item Generalized Modulo \adcs \cite{Florescu:2022:J,Florescu:2022:Cb,Florescu:2022:C}, as well as,  
\item Time-Encoding Architectures \cite{Florescu:2022:Ja,Florescu:2022:Ca,Florescu:2021:C}. 
\end{enumerate}
Let $\lambda>0$ denote the saturation threshold of the modulo \adc (\madc). Hardware experiments \cite{Bhandari:2021:J} have established that inputs as large as $25\lambda$ to $30\lambda$ can be recovered in the presence of non-idealities, system noise and quantization. This was achieved by a Fourier-domain recovery algorithm referred to as the \emph{Fourier-Prony} or \usfp method. The \usfp method is
\begin {enumerate*} [label={\alph*})]%
  \item non-iterative
  \item agnostic to $\lambda$
  \item operates at the tightest possible sampling rates
  \item is empirically robust to system noise \cite{Bhandari:2021:J,Bhandari:2022:J} and outliers \cite{Bhandari:2022:C},
\end{enumerate*}
thus making it particularly attractive for real-world experiments. For theoretical guarantees in the case of noise for the nonideal scenario, the generalized modulo \adc and associated sampling model introduce an additional parameter called \emph{hysteresis} which was used in the context of uniform sampling \cite{Florescu:2022:J} and time encoding \cite{Florescu:2022:Ja}.

\bpara{Related Work on Unlimited Sampling.} Several follow-up papers, summarised below, have adopted the \usf strategy for modulo sensing and reconstruction \eg \cite{Ordentlich:2018:J,Romanov:2019:J,Rudresh:2018:C,Musa:2018:C,Prasanna:2021:J,Ordonez:2021:J,Gong:2021:J,Azar:2022:C,Weiss:2022:J}, predominantly focused on bandlimited signals (barring \cite{Bhandari:2020:C,Bhandari:2022:J})

Previous research on bandlimited signals on a real line \cite{Bhandari:2017:C,Bhandari:2020:Ja}, periodic\footnote{We remind the reader that the transform domain separation principle in \cite{Bhandari:2021:J} is also applicable to aperiodic signals. The periodicity assumption leads to an efficient implementation via the Fast Fourier Transform (FFT) algorithm.} bandlimited signals \cite{Bhandari:2021:J} and sinusoidal tones \cite{Shtendel:2022:C} (see \fig{fig:LS}). Ordentlich \etal studied modulo sampling via rate-distortion theory \cite{Ordentlich:2018:J}, and also proposed a recovery method \cite{Romanov:2019:J} assuming a number of unfolded samples known \emph{a priori}.  
Gong \etal \cite{Gong:2021:J} proposed a multi-channel method with complex-valued moduli to recover HDR inputs.  
Azar \etal \cite{Azar:2022:C} extended the Fourier-Prony approach \cite{Bhandari:2021:J} with additional time-domain constraints. A robust algorithm based on optimization was presented in \cite{Guo:2023:C} which was validated via hardware experiments.

\bpara{Motivation for Bandpass Sampling Theory.} The topic of \emph{Bandpass Sampling} has been studied in detail from the perspective of sampling theory \eg \cite{Brown:1983:J,Vaughan:1991:J,Wahab:2022:J} as well as application-centric contexts, \eg sonar \cite{Grace:1968:J,Knight:1981:J}, radar systems \cite{Waters:1982:J,Rice:1982:J}, ultrasound imaging \cite{Kang:2022:J} and optics \cite{Hirabayashi:2002:J,Larkin:1996:J}. 

Since bandpass functions are also bandlimited, a natural first thought is to sample them at their Nyquist rate. However, this requires the \adc to operate at very high sampling rates \cite{Vaughan:1991:J}, which causes implementation challenges and imply higher power consumption in the \adc \cite{Wahab:2022:J}. Furthermore, this simplistic view does not leverage the fact that a bandpass signal contains no information in the lowpass region. 

\ipara{Existing Strategies for Bandpass Sampling.} 
Researchers developed the following strategies for efficient acquisition. 
\begin{enumerate}[leftmargin = 30pt, label = $\bullet$]
  \item Hardware-only approaches \cite{Fazi:1997:C} use mixers to demodulate the signal so that it is concentrated in the baseband region. This leads to lower sampling rates but results in complex and more expensive hardware implementation \cite{Wahab:2022:J}. 
  \item Algorithmic approaches achieve an equivalent of \emph{hardware demodulation} (above) by aliasing the bandpass signal into its passband via undersampling \cite{Knight:1981:J,Brown:1983:J,Vaughan:1991:J,Wahab:2022:J}. This is schematically explained in \fig{fig:BPsamp} and the focus of our work.
\end{enumerate}

 \begin{figure}[!t]
\centering
\includegraphics[width =0.65\textwidth]{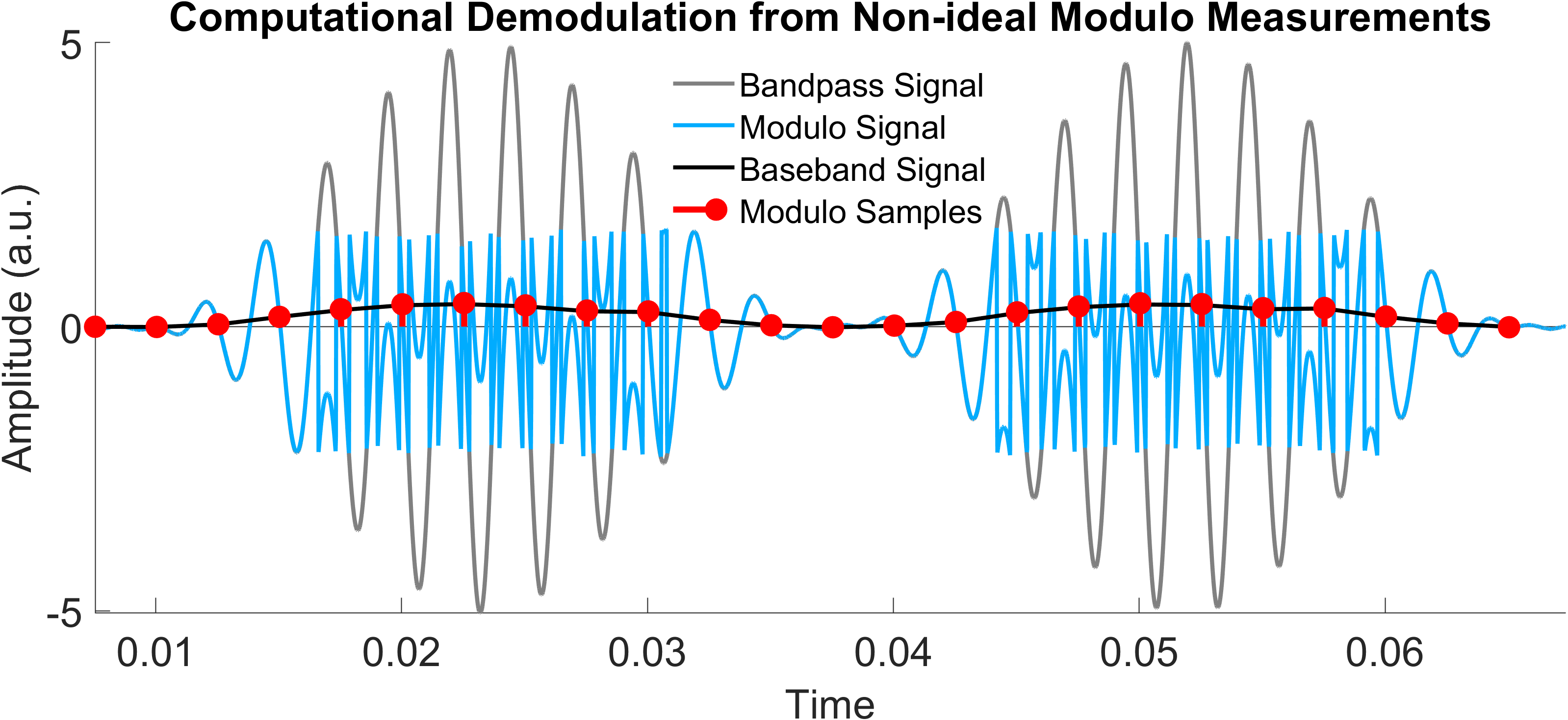}
\caption[width =1\textwidth]{An amplitude modulated signal, its associated modulo signal and undersampled modulo samples. The frequencies of the message and carrier signals are $f_{\mathsf{M}}=35\ \mathrm{Hz}$ and $f_{\mathsf{C}}=400\ \mathrm{Hz}$, respectively. For a sampling period $\ts=2.5\ \mathrm{ms}$, modulo samples coincide with underlying baseband signal samples, which directly results in demodulation.}
\label{fig:exp0}
\end{figure}

\ipara{Unlimited Sampling of Bandpass Signals.} The widespread prevalence of bandpass signals, also studied in the hardware context necessitates, the development of associated HDR \adcs  (cf.~\cite{Fazi:1997:C,Kappes:2003:J,Seo:2003:J}). Such approaches do not leverage the co-design approach \eg \usf. From the computational sensing perspective, bandpass signals have not yet been incorporated in \usf, yet they are highly relevant in known and emerging applications such as \usf based radar \cite{Feuillen:2022:C,Feuillen:2023:C}.

On the one hand, modulo folding renders a bandlimited signal, non-bandlimited \cite{Bhandari:2017:C} so it is expected that undersampling in the spirit of \cite{Brown:1983:J,Vaughan:1991:J,Wahab:2022:J,Waters:1982:J,Rice:1982:J,Hirabayashi:2002:J,Kang:2022:J}, is likely to be a theoretical dead-end. On the other hand, in our hardware experiments with \madc and amplitude-modulated (AM) bandpass signals, we have observed something surprising; in certain settings, undersampling an AM signal with \madc allows us to directly access its baseband version \emph{without} any need for recovery algorithms. A hardware example with our \madc depicting this scenario is shown in \fig{fig:exp0}. The preservation of the bandpass spectral features in this undersampled scenario, even though the energy of the modulo signal is distributed over infinite frequency support as illustrated in \fig{fig:MDP}, indicates that there may be a certain structured form of aliasing. This observation raises the following questions 

\begin{figure}[t]
    \centering
    \includegraphics[width = 0.55\columnwidth]{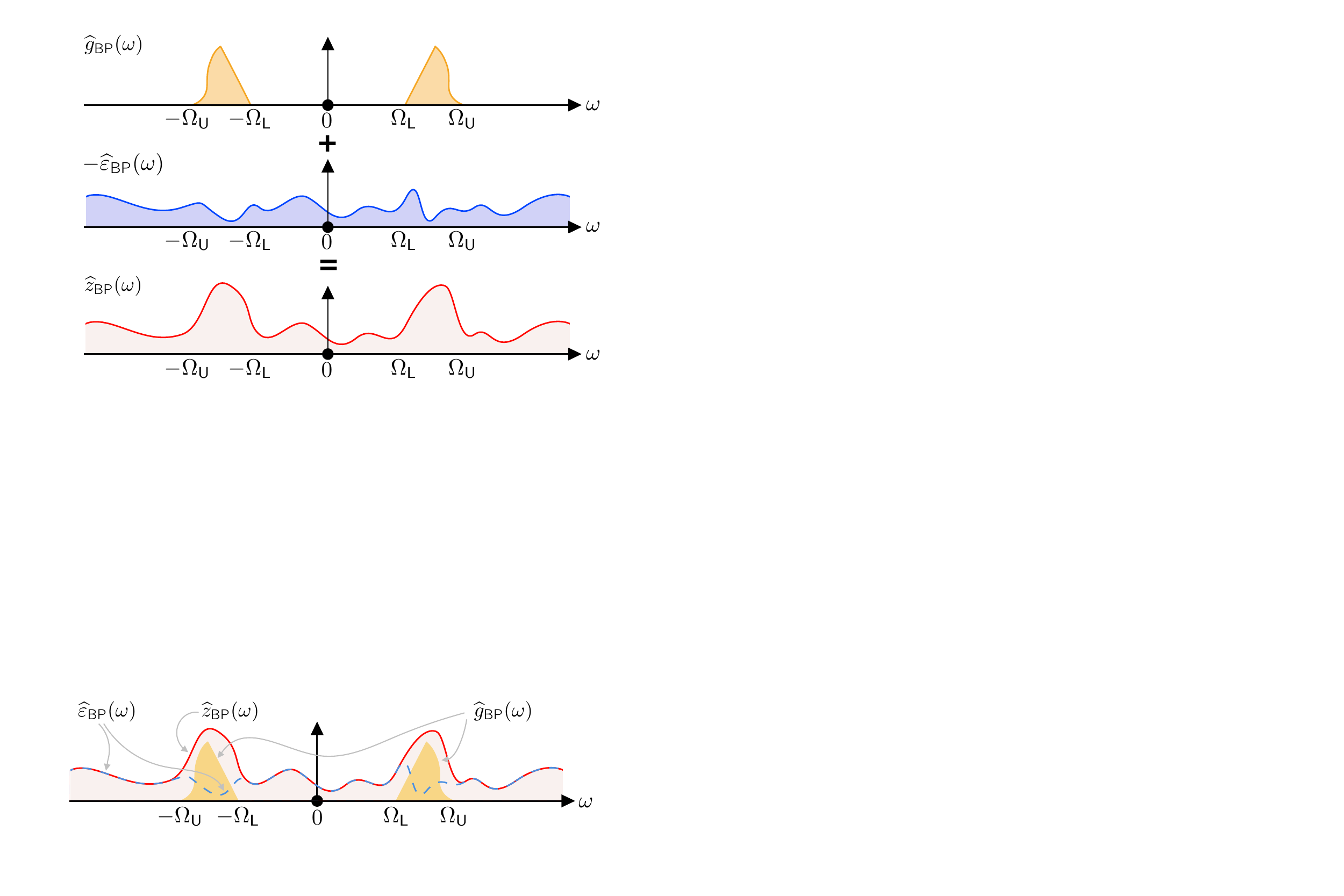}
    \caption{The Fourier-domain representation of a modulo signal. The sum of a bandpass component $\ftsub{g}{\tbp}{\omega}$ (upper) and a residual component $-\ftsub{\varepsilon}{g_{\tbp}}{\omega}$ (middle) composes the modulo signal
    $\ftsub{z}{\tbp}{\omega}=\ftsub{g}{\tbp}{\omega}- \ftsub{\varepsilon}{g_{\tbp}}{\omega}$ (lower).}
    \label{fig:MDP}
\end{figure}

\begin{enumerate}[leftmargin=40pt, label = $\textrm{Q}_{\arabic*})$]
  \item Under what conditions one can recover a bandpass signal from undersampled, folded samples?
  
\item Can the undersampling strategy be applied to different signal folding architectures arising in practice?
  
  \item Can the related recovery algorithms be validated using real experiments with the \madc?
\end{enumerate}

\bpara{Contributions.} In this work, we formulate and address the problem of HDR sampling and recovery of bandpass signals with modulo non-linearities, especially with an emphasis on undersampling. Our main contributions are as follows,
\begin{enumerate}[label = $\arabic*)$,leftmargin=30pt,itemsep=0pt]
    \item Given folded samples, we provide recovery guarantees for bandpass signals acquired at low sampling rates that enable sub-Nyquist sampling (see \fig{fig:LS}).
    
    \item By considering different modulo folding architectures that explain hardware observations, namely, (a) ideal, (b) generalized modulo non-linearities, and (c) non-ideal architectures, our work offers algorithmic flexibility for practical solutions. 
    
    \item Beyond numerical simulations, we validate our theoretical results and algorithms via hardware experiments, thus closing the gap between theory and practice. This also establishes the empirical robustness of our recovery methods. 
\end{enumerate}
The far-reaching goal of our work is to enable HDR computational sampling of bandpass signals, which arise in a number of applications \cite{Grace:1968:J,Knight:1981:J,Waters:1982:J,Rice:1982:J,Kang:2022:J,Hirabayashi:2002:J,Larkin:1996:J} including \usf-Radar \cite{Feuillen:2022:C,Feuillen:2023:C}. Compared with the bandlimited signals, our recovery guarantees are summarized in \fig{fig:LS}.

\section{Signal Folding Architectures} 
\label{sec:SmpPipe}
In this section, we revisit the ideal \cite{Bhandari:2017:C,Bhandari:2020:Ja}, generalized \cite{Florescu:2022:J} and non-ideal  modulo operators \cite{Bhandari:2021:J}, motivated from practical scenarios and hardware experiments.  We also summarize the corresponding recovery algorithms that will be re-purposed to handle bandpass signals, thus increasing the applicability to a wider range of signal folding architectures.       

\subsection{Ideal Modulo Folding \texorpdfstring{$(\mathscr{M}_\lambda)$}{Ml}}

\bpara{Acquisition.} In the ideal case, folded samples are acquired using the centered modulo mapping defined by,
\begin{equation}
\label{eq:mod}
\mathscr{M}_{\lambda}:g \mapsto 2\lambda \left( {\fe{ {\frac{g}{{2\lambda }} + \frac{1}{2} } } - \frac{1}{2} } \right), \quad \fes{g} \DE g - \flr{g}
\end{equation}
where $\flr{g}$ defines the floor function\footnote{Greatest integer less than or equal to $g$ \ie, 
$\flr{g} = \max \left\{ {k \in \mathbb{Z} \mid k \leqslant g} \right\}$.}. 
Even when $g \gg \lambda$, the range of $\mathscr{M}_\lambda$ is bounded or $\MO{g\rob{t}} \in [-\lambda, \lambda)$. For a sampling period $\ts >0$, we denote modulo samples by $y\sqb{k}=\MO{g\rob{k\ts}}$ while $\gamma\sqb{k}= g\rob{k\ts}$. A hardware prototype of the modulo \adc was presented in \cite{Bhandari:2021:J}.

\bpara{Reconstruction Strategy}. From the modulo decomposition, 
\begin{equation}
\label{eq:MDP}
g\rob{t} = \MO{g\rob{t}} + \VO{g}\rob{t}, \quad \VO{g} \in 2\lambda\Z
\end{equation}
we can deduce that unfolding $y\sqb{k}$ is equivalent to recovering $r\sqb{k} = \VO{g}\rob{k\ts}\in 2\lambda\Z
$ (residue) which is a simple function \cite{Bhandari:2020:Ja}. This recovery approach was first introduced in the Unlimited Sampling Algorithm or \usalg \cite{Bhandari:2020:Ja}. Key to the recovery in \usalg is the idea that $r\sqb{k}$ can be extracted from $y\sqb{k}$. For any $\BL{g}$, its samples are highly correlated and hence, $\|\Delta^N \gamma\|_\infty \leq {\left( {\ts \Omega \e} \right)^N} \| \gamma\|_\infty$. This implies that oversampling shrinks the higher order differences, $\rob{\Delta^{N}  \gamma}\sqb{k}$. Choosing,  
\begin{equation}
    \ts \leq \TUS = \frac{1}{2\Omega e} \mbox{ and } N^\star \geq \left\lceil {\frac{{\log \lambda  - \log \beta_g}}{{\log \left( {\ts\Omega \e} \right)}}} \right\rceil
    \label{eq:TUS}
\end{equation}
with $\beta_g\geq \maxn{g}, \beta_g \in 2\lambda\Z$ it is guaranteed that $\maxn{\Delta^{N^\star}  \gamma}< \lambda$ \cite{Bhandari:2020:Ja}. 
For $\{\ts,N^\star\}$ above, it follows that $\Delta^{N^\star} \gamma = \MO{\Delta^{N^\star} \gamma} = \MO{\Delta^{N^\star} y}$ or  $\Delta^{N^\star}\varepsilon_g = \MO{\Delta^{N^\star} y} - \Delta^{N^\star} y  $. Thereon, stable inversion of $\Delta^{N^\star}$ is facilitated by leveraging the restriction on the residual function's amplitudes or $\VO{g}\rob{kT} \in 2\lambda\Z$ and growth properties of $\BL{g}$. Finally, $g\rob{t}$ is obtained by interpolating $\gamma\sqb{k} = y\sqb{k}+r\sqb{k}$. 

\subsection{Generalized Modulo Nonlinearity \texorpdfstring{$(\MOh)$}{Mlh}}
A generalized modulo architecture or $\MOh$ \cite{Florescu:2021:C,Florescu:2022:J,Florescu:2022:Cb,Florescu:2022:C, Florescu:2022:Ca} was proposed to tackle fundamental hardware non-idealities \eg \emph{hysteresis} which is present during \adc quantization. This means that the threshold is different when the input signal is increasing or decreasing. Since modulo samples can be interpreted as \emph{quantization noise} relative to the continuous input \cite{Bhandari:2020:C}, it was argued that the hysteresis effect is also present in \madcs \cite{Florescu:2022:Cb}. Including hysteresis in $\MOh$ is beneficial to the model, allowing to tackle non-standard sampling scenarios such as incomplete folds \cite{Florescu:2022:J}, non-pointwise average sampling \cite{Florescu:2022:C}, or event-driven sampling \cite{Florescu:2021:C,Florescu:2022:Ca}.

\bpara{Acquisition.} We define the $\MOh$ as follows.
\begin{definition}[Generalized Modulo]
\label{def:GNM}
	The analog modulo encoder with threshold $\lambda$, hysteresis parameter $h\in\left[0,2\lambda\right)$ and transient parameter $\alpha$, is an operator $\MOh:L^2\rob{\mathbb{R}} \rightarrow L^2\rob{\mathbb{R}}$, $\boldsymbol{\mathsf{H}}=\sqb{\lambda,h,\alpha}$, that generates an analog function $z \rob{t}=\MOh g \rob{t},$ $ t\geq \tau_0$ in response to input $g \in \PW{\Omega}$ defined by
	\begin{equation*}
		z \rob{t}=g\rob{t}-\varepsilon_{\boldsymbol{\mathsf{H}}}\rob{t},
	\end{equation*}
	where 
$${\varepsilon_{\boldsymbol{\mathsf{H}}}\rob{t}}=\sum\limits_{p\in\mathbb{Z}} s_p \varepsilon_0\rob{t-\tau_p},\ t\in\mathbb{R}$$ 
is known as the residual, $$\varepsilon_0\rob{t}=2\lambda_h \sqb{ \rob{1/\alpha} t\cdot \ind_{[0,\alpha)}\rob{t}+ \ind_{[\alpha,\infty)}\rob{t}},$$ $\lambda_h=\lambda-h/2$, $s_p= \mathrm{sign} \rob{g\rob{\tau_p}-g\rob{\tau_{p-1}}}, p\geq1$ and  $\cb{\tau_p}$ is an asynchronous sequence satisfying $\tau_0=0$ and
	\begin{align*}
		\tau_1&=\min\cb{t > \tau_0 \vert \MO{g\rob{t}+\lambda}= 0 }\\		\tau_{p+1}&=\min\cb{t> \tau_p \vert \MO{g\rob{t}-g\rob{\tau_p}+h s_p}=0}.
	\end{align*}
\end{definition}

The $\MOh$ operator was shown to accurately describe the experimentally observed output of the \madc for different values of the design parameter $h$\footnote{To get a sense of the hardware functionality, we refer the reader to the live hardware demonstration here: \href{https://youtu.be/R4rji5jOjD8}{\texttt{https://youtu.be/R4rji5jOjD8}}.}.

\bpara{Reconstruction Strategy.} The typical assumption is that the input belongs to  $\LL{2}\cap \mathcal{B}_{\Omega}$, also known as the Paley-Wiener space $\PW{\Omega}$. We filter the data with $\psi_N$, which leads to 
\begin{equation*}
	\label{eq:psi_y}
	\psi_N \ast y\sqb{k}=\psi_N\ast\gamma\sqb{k} - \psi_N\ast r\sqb{k},
\end{equation*}
where $r[k]=\varepsilon_{\boldsymbol{\mathsf{H}}}(kT)$ and
$\psi_N $ can be the finite difference filter $\Delta^N$ \cite{Florescu:2022:J} or a generalized filter allowing additional noise robustness \cite{Florescu:2022:Cb}. The general idea  is to compute $r\sqb{k}$ and then $\gamma\sqb{k}=y\sqb{k}+r\sqb{k}$. Thus, the only unknowns are the folding times $\tau_p$ and signs $s_p$.
Generally, when the transient duration $\alpha$ is not negligible, the recovery is performed via \emph{thresholding}, exploiting that $\psi_N\ast r\sqb{k}$ consists of pulses centered in each folding time, with shape given by  $\zeta_N\sqb{k}=\sum_{i=-\infty}^k \psi_N\sqb{i}$, the step response of $\psi_N$. Thresholding requires that $\psi_N\ast\gamma\sqb{k}$ is small and that the folding times are separated so that the pulses corresponding to different folds do not overlap. This can be enabled via the hysteresis parameter $h$ allowing a lower bound on the spacing between folds $    \tau_{p+1}-\tau_p \geq \frac{h}{\Omega \norm{g}_\infty}$. 
For $\psi_N=\Delta^N$, the thresholding recovery conditions are \cite{Florescu:2022:J}
\begin{enumerate}[leftmargin = 40pt]
\item $\rob{T\Omega e}^N \norm{g}_\infty + 2^N\eta_\infty\leq\frac{\lambda_h}{2N}$,
\item $\rob{N+1} T\Omega \norm{g}_\infty \leq h^\ast$,\quad $h^\ast=\min\cb{h,2\lambda-h}$.
\end{enumerate}
If $\alpha$ is negligible, then the recovery can also be performed with \usalg, as will be detailed in Section \ref{sect:genmodulo_recovery}.

\subsection{Non-Ideal Modulo Folding \texorpdfstring{$(\widetilde{\mathscr{M}}_\lambda)$}{Mnil}}
In practice, despite its simplicity, $\VO{g}$ satisfies $\VO{g} \not\in 2\lambda\Z$. In that case, any algorithm that assumes $\VO{g} \in 2\lambda\Z$ will fail to recover the residue signal\footnote{In the example of the \usalg, $\BL{g}$ and $\VO{g}\not\in 2\lambda\mathbb{Z}$ can no longer be separated by $\Delta^N$ operation because $\VO{g}\not\in 2\lambda\mathbb{Z} \Rightarrow \Delta^{N} \gamma \not= \MO{\Delta^{N^\star} y}$.}. This required flexible algorithms that are agnostic to $\lambda$ and can handle $\VO{g} \not\in 2\lambda\Z$.

\bpara{Acquisition.} The non-ideal modulo operator $\MONI{\cdot}$ \cite{Bhandari:2021:J} is implicitly defined via $g\rob{t} = \MONI{g\rob{t}}+\mathscr{R}_g\rob{t}$ where,
\begin{equation}
\label{eq:res}
\mathscr{R}_g\rob{t}= \sum\limits_{m \in \mcal{M}} {c_m{\ind_{{\mathcal{D}_m}}}\left( t \right)}, 
\ \ 
\cupdot_m \ind_{{\mathcal{D}_m}} = \R,
\ \ 
c_m\iR.
\end{equation}
$\mcal{M}$ is the set of folding instants, $\mathcal{D}\subseteq\mathbb{R}$ and unlike $\VO{g}$ in \eqref{eq:MDP}, $\mathscr{R}_g$ can take arbitrary real values. The non-ideal residue in \eqref{eq:res} can describe (a) the effect of  reset- or shot-noise disruptions on the sampled signal\cite{Bhandari:2022:C}, (b) imperfections in the acquisition device \cite{Florescu:2022:C} and, (c) deviations in the modulo threshold. 

\bpara{Reconstruction Strategy}. To recover $\BL{g}$ from $y\sqb{k} = \MONI{g\rob{k\ts}}$, the ``Fourier-Prony'' algorithm or \usfp was devised in \cite{Bhandari:2021:J}. The \usfp seeks to recover $r\sqb{k} = \mathscr{R}_g \sqb{k}$ in the Fourier-domain while making no assumptions about $\lambda$. Instead, the spectral properties of $y\sqb{k}$ and the $\BL{g}$ are exploited to achieve a Fourier-domain separation between $\gamma\sqb{k}$ and $r\sqb{k}$%
\footnote{Note that $\BL{g}$ but $ \MONI{g\rob{t}} \notin \mathcal{B}_{\Omega}$. The spectrum of  $\MONI{g}$ can be split into an in-band component $\rob{\BL{g}}$ and an out-of-band component that only contains $ \mathscr{R}_g$. Furthermore, the parametric structure of $\mathscr{R}_g$ with the number of parameters depending on the folding instances, dictates the number of out-of-band accessible values that are required for the recovery of  $r\sqb{k}$.}.
It was shown that using the \usfp algorithm, a $\tau$-periodic, $\Omega$-bandlimited signal folded at most $M$ times can be exactly recovered (up to a constant) from its modulo samples if the sampling period obeys,
\begin{equation}
   \ts \leq \TFD = \frac{\tau}{K}
   \quad 
   \mbox{where}
   \quad 
   K \geq 2\rob{\cil{\frac{\Omega\tau}{2\pi}}+M+1}.
   \label{eq:tfd}
\end{equation}
With $\ts$ above, spectral estimation techniques are harnessed by \usfp to robustly recover $r\sqb{k}$ and hence $g\rob{t}$.  

\section{Towards Computational Demodulation via Modulo Undersampling}
\label{sec:BPwithOL}
In the following, we consider a real-valued bandpass signal $\BP{g_\tbp} \subseteq \mathcal{B}_{\Omega_{\mathsf{U}}}$, where $0<\ol<\ou$. We can use \eqref{eq:MDP} to write $
z_{\tbp}\rob{t} = \MO{g_{\tbp}\rob{t}} = 
g_{\tbp}\rob{t} - \VO{g_{\tbp}}\rob{t}$. The samples of $z_{\tbp}\rob{t}$ at intervals of $\ts = 2\pi / \os$ are
\begin{equation}
    y_{\tbp}\sqb{k} \DE z_{\tbp}\rob{t}|_{t=k\ts} =  \overbrace{g_{\tbp}\rob{t}|_{t=k\ts}}^{ = \gamma_{\tbp}\sqb{k}} - \overbrace{\VO{g_{\tbp}}\rob{t}|_{t=k\ts}}^{ = r_{\tbp}\sqb{k}}.
    \label{eq:yBP}
\end{equation}
Since $g_{\tbp}\rob{t} \in \mathcal{B}_{\Omega_{\mathsf{U}}}$, we can reinterpret a bandpass signal as a bandlimited signal with bandwidth $\ou$ and thus, the recovery methods of the USF apply. A naive reconstruction approach to obtain $g_{\tbp}\rob{t}$ from modulo samples $y_{\tbp}\sqb{k}$ would entail working with the three recovery models in Section \ref{sec:SmpPipe} and the corresponding sampling periods given in \fig{fig:LS}. Since this approach does not exploit $\ftsub{g}{\tbp}{\omega}=0$, $\forall \omega \in \sqb{-\ol, \ol}$, we observe that the corresponding USF sampling periods are exorbitant, and often impractical in real-world scenarios.

\begin{figure}[!t]
    \centering
     \scalebox{0.65}{\input{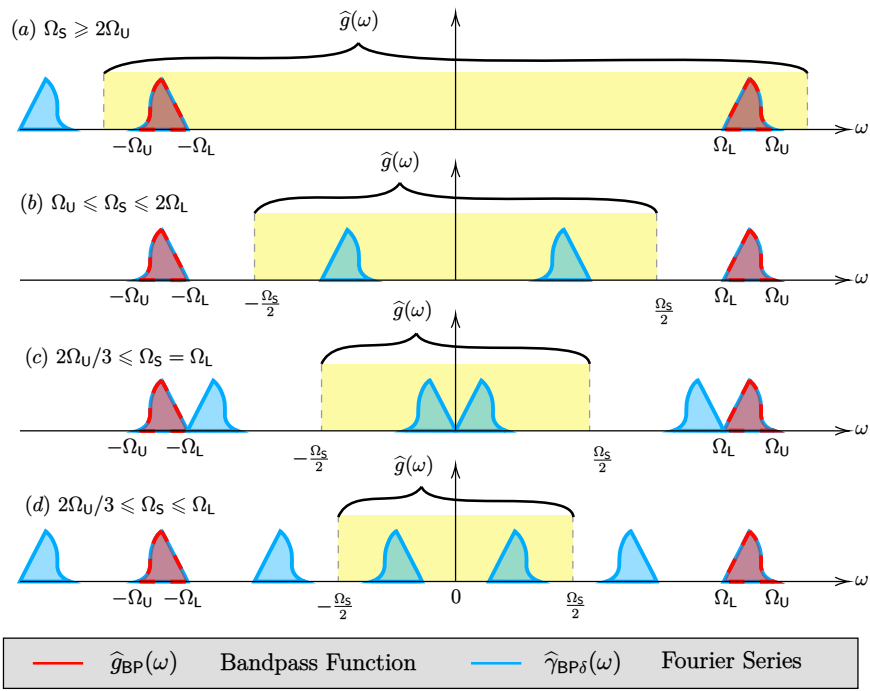}}  

    \caption{The computation of the baseband function $g(t)$. The Fourier transform of the bandpass function (red), and periodicity of the Fourier series of the bandpass samples (blue)  at four sampling frequencies. (a) Sampling above the critical rate of the bandpass signal. (b) Sampling frequency within the wedge of $P=2$. (c) Maximal valid sampling frequency for $P=3$. (d) Sampling frequency within the wedge of $P=3$.}
    \label{fig:BPsamp}
\end{figure}

This paper focuses on a \emph{computational sensing} framework for efficient sampling of folded bandpass signals. Before we discuss the context of modulo samples, we discuss the conventional setup \cite{Brown:1983:J,Vaughan:1991:J,Wahab:2022:J} and revisit the concept of \emph{sampling-based-demodulation}. In this case, the natural interplay between the following two entities is leveraged, 
\begin{itemize}
  \item Fourier-domain bandpass structure.
  \item Spectral shifting by controlling sampling rate.
\end{itemize}
Sampling a signal at period $\ts$ replicates its spectral content on intervals with period $\os$. If $\os$ is smaller than the critical period of $2\ou$, the replicas alias inside $\rob{-\ol, \ol}$, and a corresponding lowpass signal can be extracted. We qualitatively demonstrate this for $4$ exemplary periods in \fig{fig:BPsamp}. When $\os \leq 2\ou$, as in \fig{fig:BPsamp}(b)-(d), the lowpass signal appears within the zero-centered band of $\rob{-\ol,\ol}$. Notice the spectrum mirroring effect in \fig{fig:BPsamp}(b). For recovery to be possible, the aliased replicas must not overlap. 
The conditions for non-overlapping spectral lobes are formalized in the next lemma, also known as the bandpass sampling theorem \cite{Vaughan:1991:J}.

\begin{lemma}
\label{lem:1}
Let $\BP{g_{\tbp}}$ be a continuous-time bandpass signal and let $\gamma_{\tbp}\sqb{k}$ be its uniform samples with $\ts \leq \pi / \blb$ where $\blb \DE \ou -\ol$. Define $\ftsub{\gamma}{\tbp\delta}{\omega} =\ts \sum_{k \in \Z} \gamma_{\tbp}\sqb{k} e^{-\jmath \omega k \ts}$.
\noindent If there exists $\cndp \in \N$ such that,
\begin{equation}
\label{eq:bps_cond}
\frac{\pi\rob{\cndp-1}}{\ol} \leq \ts \leq \frac{\pi\cndp}{\ou},
\end{equation}
then $\ftsub{\gamma}{\mathsf{BP}\delta}{\omega} =\ftsub{g}{\mathsf{BP}}{\omega}, \omega \in \left(-\ou, -\ol\right) \cup \left(\ol, \ou\right)$ and $g_{\mathsf{BP}}\rob{t}$ can be recovered from $\gamma_{\tbp}\sqb{k}$.
\label{lem:BPS}
\end{lemma}
In Lemma \ref{lem:1}, $\cndp$ is a integer constant such that there are $\cndp-1$ periods within $\rob{-\ol , \ol}$, i.e., $\flr{\tfrac{2\ol}{\os}}=\cndp-1 \in \N$.
For $\os$ and $\cndp$ satisfying Lemma \ref{lem:BPS} and given  $\gamma\sqb{k} \DE g\rob{k\ts}$, we define $\ftsub{\gamma}{\delta}{\omega} \DE \ts \sum_{k \in \Z} \gamma\sqb{k} e^{-\jmath \omega k \ts}$. Then $g_{\tbp}\rob{t}$ can be obtained by computing the inverse FT of,
\begin{align}
\label{eq:Fshift}        
&    \ftsub{g}{\tbp}{\omega}  = \mathbbm{1}_{\mathcal{D}\rob{ \os,\cndp}}\rob{\omega} \ftsub{\gamma}{\delta}{\omega}, \mbox{ where}  \\
&    \mathcal{D}\rob{ \os,\cndp}  = \tfrac{1}{2}{\left[ - {\cndp}\os, -{(\cndp-1)}\os\right]
\cup 
\tfrac{1}{2}{\left[{(\cndp-1)}\os, {\cndp}\os\right]}}. \notag
\end{align}

\bpara{From Spectral Overlap to Bandpass Recovery.} Sampling-based-demodulation is designed so that bandpass content occupies the vacant lowpass regions via aliasing. Adapting this strategy to the \usf is highly non-trivial as $z_{\tbp}\rob{t} \notin \mathcal{B}_{\rob{\ol,\ou}}$. The relation between the bandpass and the residual components is shown in \fig{fig:MDP}, and reveals that for modulo bandpass signals there is no interval on which $\ftsub{z}{\tbp}{\omega}= \ftsub{g}{\tbp}{\omega} - \ftsub{\varepsilon}{g_{\tbp}}{\omega}=0$. Therefore, previous approaches do not apply\cite{Brown:1983:J,Vaughan:1991:J,Wahab:2022:J}. 
Nonetheless, the example in \fig{fig:exp0} implies that spectral overlapping does not necessarily prevent demodulation, albeit in the context of bandpass signals. 
On the flip side, we show that one can demodulate even non-bandpass signals such as $\ftsub{z}{\tbp}{\omega}$ by sampling.
The goal of this paper is to develop mathematical guarantees and frameworks for computational demodulation of modulo signals. By looking at the two components of a modulo signal \eqref{eq:MDP}, the task is to determine the intersection between sampling periods that, 
\begin{enumerate}[leftmargin=30pt,itemsep=0pt]
  \item do not create spectral overlap of $g_{\tbp}\rob{t}$, and,
  \item are sufficient to recover the residual samples.
\end{enumerate}

\subsection{Undersampling with the Ideal Modulo Architecture}
\label{sec:idealMUS}

In what follows, we prove a sampling theorem enabling ideal modulo sampling of bandpass signals with a sampling rate that improves the known result $\ts \leq \TUS = \nicefrac{1}{2\Omega e} $ \eqref{eq:TUS}.  
\begin{theorem}[The Unlimited Bandpass Sampling Theorem]
\label{thm:USBP}
Let $g_{\tbp}\rob{t} \BP{}$ and $z_{\tbp}\rob{t}=\MO{g_{\tbp}\rob{t}}$. Let $y_{\tbp}\sqb{k} = z\rob{k \ts}$, $k \in \Z$ be the modulo samples of $g_{\tbp}\rob{t}$ at period $\ts > 0$.
If there exists $\cndp \in \N$ such that $\ts$ satisfies,
\begin{equation}
\begin{cases}
\frac{\pi\rob{\cndp-1}}{\ol} \leq \ts \leq \frac{2\pi e \rob{\cndp-1}+1}{2 e\ou} & \cndp \in 2\N+1  \\
\frac{2\pi e\cndp-1}{2 e\ol} \leq \ts \leq \frac{\pi \cndp}{\ou}
&\cndp \in 2\N
\end{cases}              
\label{eq:th_USBP}
\end{equation}
where e denotes Euler’s number, then $g_{\tbp}\rob{t}$ can be recovered from $y_{\tbp}\sqb{k}$ up to an integer multiple of $2\lambda$.
\label{theo:USBP}
\end{theorem}

\begin{proof}
We define the normalized FS with coefficients $y_{\tbp}\sqb{k}$,
\begin{align}
 \label{eq:FSBP}
    \ftsub{y}{\tbp\delta}{\omega} 
     &\DE \ts \sum_{k \in  \Z} y_{\tbp}\sqb{k}e^{-\jmath \omega k \ts}\\
    &\EQc{eq:yBP} \ts \sum_{k \in  \Z} \left[ \gamma_{\tbp} - r_{\tbp} \right]\sqb{k} e^{-\jmath \omega k \ts}
     = \sqb{\hat{\gamma}_{\tbp\delta}-\hat{r}_{\tbp\delta}} \rob{\omega}, \notag
\end{align}
and the lowpass filtered signal $g \rob{t}$ via its spectrum,
\begin{equation}
\label{eq:LPF}
   \ft{g}{\omega}=\mathbbm{1}_{\left[ - \tfrac{\os}{2}, \tfrac{\os}{2} \right]}\rob{\omega} \ftsub{\gamma}{\tbp\delta}{\omega}.
\end{equation}
Using that $\ts \gamma_{\tbp}\sqb{k}$ are the FS coefficients of $\ftsub{\gamma}{\tbp\delta}{\omega}$ together with \eqref{eq:LPF} we write,
\begin{align}
\label{eq:gbpg}
         \gamma_{\tbp}\sqb{k} &= \frac{1}{2\pi} \int_{-\os/2 }^{\os/2} \ftsub{\gamma}{{\tbp}\delta}{\omega} e^{\jmath \omega k \ts}d\omega  \\
    &= \frac{1}{2\pi}  \int_{-\infty}^{\infty} \mathbbm{1}_{\left[ - \tfrac{\os}{2}, \tfrac{\os}{2} \right]}\rob{\omega} \ftsub{\gamma}{\tbp\delta}{\omega} e^{\jmath \omega k \ts}d\omega  \notag \\
    &\equalref{ \eqref{eq:LPF}} \frac{1}{2\pi}  \int_{-\infty}^{\infty} \ft{g}{\omega} e^{\jmath \omega k \ts}d\omega  = g\rob{t}_{|t=k\ts} \DE \gamma \sqb{k}.\notag
\end{align}
We can use \eqref{eq:MDP} to write $g\rob{t} = z\rob{t} + \VO{g}\rob{t}$, where $z\rob{t}=\MO{g\rob{t}}$. Respectively, we write the samples $\gamma\sqb{k}$ as,
\begin{equation}
    \gamma\sqb{k}= y\sqb{k} + r\sqb{k}.
    \label{eq:y}
\end{equation}
The difference between the two discrete modulo sequences is,
\begin{equation}
     y_{\tbp}\sqb{k} -  y\sqb{k} = \underbrace{\gamma_{\tbp}\sqb{k} -\gamma\sqb{k}}_{=0 \text{ from \eqref{eq:gbpg}}} + r\sqb{k} -  r_{\tbp}\sqb{k}.
     \label{eq:grid}
\end{equation}
For the LHS of \eqref{eq:grid}, since $y_{\tbp}\sqb{k} \in [-\lambda,\lambda)$ and $-y_{\tbp}\sqb{k} \in (-\lambda,\lambda]$, we get that $y_{\tbp}\sqb{k} - y\sqb{k} \in (-2\lambda,2\lambda)$. For the RHS, $r\sqb{k} -  r_{\tbp}\sqb{k} \in 2\lambda\Z$. Hence, the intersection set between the two sides of \eqref{eq:grid} must be zero and we get,
\begin{equation}
\label{eq:yBPtoy}
        y_{\tbp}\sqb{k} -  y\sqb{k} =  r\sqb{k} -  r_{\tbp}\sqb{k}  = 0  \Rightarrow  y_{\tbp}\sqb{k} =  y\sqb{k}.
\end{equation}
Substituting \eqref{eq:gbpg} and \eqref{eq:yBPtoy} into \eqref{eq:FSBP} we get,
\begin{equation}
 \label{eq:FSBP2}
    \ftsub{y}{\tbp\delta}{\omega} 
     = \ftsub{y}{\delta}{\omega} = \ftsub{\gamma}{\delta}{\omega} - \ftsub{r}{\delta}{\omega},
 \end{equation}    
where $\ftsub{\gamma}{\delta}{\omega}$ and $\ftsub{r}{\delta}{\omega}$ are the FS of $\gamma[k]$ and $r\sqb{k}$, respectively. From \eqref{eq:FSBP2}, we get that $\ftsub{y}{\delta}{\omega} = \ftsub{\gamma}{\delta}{\omega} - \ftsub{r}{\delta}{\omega}$ which means that under the conditions of the US theorem \cite{Bhandari:2020:Ja} we can get $\gamma\sqb{k}$ (and hence, $g\rob{t}$) from $y_{\tbp}\sqb{k}$. Moreover,
if \eqref{eq:bps_cond} is satisfied then from Lemma \ref{lem:BPS}, $g_{\tbp}\rob{t}$ can be obtained and
\begin{gather}
        \oup = \begin{cases}
    \ou - \frac{\cndp-1}{2}\os \quad &\cndp \in 2\N+1 \\
   \frac{\cndp}{2}\os  -\ol \quad &\cndp \in 2\N
    \end{cases},
\label{eq:USLM}
\end{gather}
where $\oup$ is the bandwidth of $g$. The USF recovery condition for $g\rob{t}$ requires that
$\os\geq 4\pi\ou^g e$. For $\cndp \in 2\N+1$, plugging \eqref{eq:USLM} allows to re-write the USF recovery condition as
\begin{equation}
\begin{aligned}
    \label{eq:USF_Omega_S}
    \os\geq 4\pi\ou^g e=4\pi e \rob{\ou-\frac{P-1}{2}\os}\\
    \Leftrightarrow \os\geq \frac{4\pi e}{1+ 2\pi e \rob{\cndp-1}}\ou.
\end{aligned}
\end{equation}
It can be directly shown that $\frac{4\pi e}{1+ 2\pi e \rob{P-1}}\ou\geq \frac{2\ou}{P}$, which yields the final condition on $\os$,
\begin{equation}
\label{eq:1}
    \frac{4\pi e}{1+ 2\pi e \rob{\cndp-1}}\ou\leq\os\leq \frac{2\ol}{\cndp-1}.
\end{equation}
For $\cndp \in 2\N$, using the same derivation as in \eqref{eq:USF_Omega_S}, we get $\os\geq 4\pi e\rob{ \frac{\cndp}{2}\os  -\ol }$, which yields 
\begin{equation}
\label{eq:2}
    \frac{2\ou}{\cndp}\leq\os\leq \frac{4\pi e}{2\pi e \cndp -1}\ol.
\end{equation}
Using $T_S=\frac{2\pi}{\os}$ in (\ref{eq:1},\ref{eq:2}) yields the desired result; then $g_{\tbp}\rob{t}$ can be recovered from its modulo samples.
\end{proof}
Theorem \ref{theo:USBP} is consistent with \eqref{eq:TUS}. Given that $g_\tbp \in \mathcal{B}_{[\ol,\ou]} \subseteq  \mathcal{B}_{\ou}$, if we sample above the critical rate we have $\cndp=1$ and thus condition \eqref{eq:1} is satisfied for $\ts \leq \TUS$. Nevertheless, for $\cndp > 1$, the bandpass prior is exploited and this results in lower sampling periods. In the bandpass result of \cite{Vaughan:1991:J}, $\cndp$ is bounded by $\cndp \leq \flr{\frac{\ou}{\blb}}$. We can find an equivalent bound for the modulo bandpass case from \eqref{eq:1} and \eqref{eq:2} as,
\begin{equation*}
\begin{cases}
(\cndp-1)2\pi e \blb \leq \ol & \cndp \in 2\N+1 \\
  \ \qquad    2\pi e \cndp \blb \leq \ou& \cndp \in 2\N
\end{cases}
\end{equation*}
which implies that,
\begin{equation}
\begin{cases}
 \cndp \leq \flr{\frac{\ol}{2\pi e \blb}+1} & \cndp \in 2\N+1 \\
 \cndp \leq \flr{\frac{\ou}{2\pi e \blb}} & \cndp \in 2\N
\end{cases}.
\label{eq:Pbound}
\end{equation}
Since the bound in \eqref{eq:Pbound} is tighter for the even case, we get that $\cndp \leq \flr{\frac{\ou}{2\pi e \blb}}$, which is stricter than $\cndp \leq \flr{\frac{\ou}{\blb}}$. The sampling frequencies from Theorem \ref{thm:USBP} that achieve minimal bandwidth $\oup=\blb$ for $g$ are
\begin{equation}
\begin{cases}
\os = \frac{2}{\cndp-1}\ol & \cndp \in 2\N+1 \\
\os =  \frac{2}{\cndp}\ou & \cndp \in 2\N
\end{cases}
\label{eq:basebandrates}
\end{equation}
which produce the \textit{baseband relocation} cases, and demonstrated in \fig{fig:BPsamp}(c) for $\cndp =3$.

\bpara{Sampling-Based Demodulation of Modulo AM Signals.} There are particular cases where the result of Theorem \ref{theo:USBP} can be improved based on an even stronger prior. For example, we can consider the case of double-sideband AM signals\cite{Pursley:2002}, which are bandpass functions whose spectrum is consisted of two adjacent sidebands that are mirrored around a carrier frequency. These signals are among the most common waveforms that are encountered in practical applications. Because of their spectral structure, AM signals can be recovered even when $\os = \blb$ \cite{Brown:1983:J}. In the next corollary, we extend this result for modulo-folded AM signals.
\begin{cor}[Unlimited Sampling of AM Signals]
\label{cor:AMUS}
If the side lobes of $\hat{g}$ are symmetric, \ie,  $\hat{g}\rob{\ol+\omega}=\hat{g}\rob{\ou-\omega},  \forall\omega\in\sqb{0,\blb/2}$, then $g_{\tbp}(t)$ can be recovered via Algorithm \ref{alg:1} if there exists $\cndpAM\in \N$ such that \begin{equation}
        \os = \frac{\ou+\ol}{2\cndpAM} 
        \quad \mbox{\textit{and}} 
        \quad \frac{\ou+\ol}{\ou-\ol} \geq 4\pi e \cndpAM.
    \label{eq:AMrates}
\end{equation} 
\end{cor}
\begin{proof}
We require that the sideband copies lead to pairs of lobes that perfectly overlap, such that a pair is centered around the origin. This is true if the LHS equation in \eqref{eq:AMrates} is satisfied. Furthermore, to ensure that adjacent pairs of lobes do not lead to additional overlaps, which ensures that the information can be recovered, we further assume $\os \geq \blb$. This leads to the following condition on $\ou, \ol$:
\begin{equation}
 \frac{\ou+\ol}{\ou-\ol} \geq 2\cndpAM.
    \label{eq:OUOLCond}
\end{equation}
Lastly, we compute $\ou^g$, in order to derive a bound for USF. The pair of lobes centered in the origin has width $\ou-\ol$, therefore its bandwidth is $\ou^g=\frac{\blb}{2}$. This yields the USF condition  $
    \os \geq 4\pi e \frac{\blb}{2} = 2\pi e \blb$. By using the LHS equation in \eqref{eq:AMrates} it is translated in terms of $\ou, \ol$  that leads to the RHS equation in \eqref{eq:AMrates}.
\end{proof}
We remark that the RHS equation in \eqref{eq:AMrates} is a tighter condition than \eqref{eq:OUOLCond}, and thus is enough for input recovery in the case of AM signals. 
As a particular case, we highlight the implication of Corollary \ref{cor:AMFP} for sinusoids, which can be viewed as AM signals with a constant modulated amplitude, $\ql$ that corresponds to the carrier frequency and $\blb \xrightarrow{} 0$. While \cite{Shtendel:2022:C} demonstrated that these signals allow lower sampling rates than in the general bandlimited case, under the conditions of Corollary \ref{cor:AMFP}, this improvement can be further enhanced.
\begin{algorithm}[!t]
\SetAlgoLined
{\bf Input:} $y_{\tbp}\sqb{k}$, $\lambda$,$\beta_{g_{\tbp}} = \maxn{g_{\tbp}}$, $\ol$, and $\ou$.\\
\KwResult{Bandpass function, $\widetilde g_{\tbp}\rob{t}$. }

\begin{enumerate}[label = $\arabic*)$,leftmargin=30pt,itemsep=1pt]
\item Compute $\oup$ using \eqref{eq:USLM}.
\item Compute $N=\left\lceil \frac{\log \lambda - \log \beta_{g_{\tbp}}}{\log \rob{T_s\oup e}} \right\rceil$.
\item Compute $\rob{\Delta ^N y}\sqb{k} \equalref{\eqref{eq:yBPtoy}} \rob{\Delta ^N y_{\tbp}}\sqb{k}$. 
\item Compute $\rob{\Delta ^N r}\sqb{k} = \rob{\MO{\Delta ^N y} - \rob{\Delta ^N y}}[k]$. \\
Set $s_{(0)}\sqb{k} = \rob{\Delta ^N r}\sqb{k}$. 
\item For $n=0:N-2$ \\
        \begin{enumerate}[label=(\roman*)]
            \item $s_{(n+1)}\sqb{k} = \mathsf{S}\rob{s_{(n)}}\sqb{k} $, \scriptsize $\mathsf{S}: \{s[k]\}_{k \in \Z^+} \rightarrow \sum_{m=1}^{k}s[m]$. \normalsize
            \item $s_{(n+1)}\sqb{k} = 2\lambda\cil{\frac{\flr{s_{(n+1)} / \lambda}}{2}}$ (rounding to $2\lambda\Z$).
            \item Compute $\kappa_{(n)} = \flr{\tfrac{\mathsf{S}^2\rob{\Delta ^n r}[1] - \mathsf{S}^2\rob{\Delta ^n r}[\tfrac{6\beta_{g_{\tbp}}}{\lambda}+1]}{12\beta_{g_{\tbp}}}+\tfrac{1}{2}}$.
            \item $s_{(n+1)}\sqb{k} = s_{(n+1)}\sqb{k}  + 2\kappa_{(n)}$.
        \end{enumerate}
    end
\item $\widetilde \gamma\sqb{k} = \mathsf{S}\rob{s_{(N-1)}}\sqb{k} + y_{\tbp}\sqb{k} + 2m\lambda$, $m \in \Z$. 
\item Compute $\ftsub{\gamma}{\delta}{\omega} = \ts \sum_{k \in \Z} \widetilde{\gamma}\sqb{k} e^{-\jmath \omega k \ts}$.
\item Estimate $\widetilde g_{\tbp}\rob{t}$ from $\ftsub{\gamma}{\delta}{\omega}$ using \eqref{eq:Fshift} and $\cndp$ from \eqref{eq:th_USBP}.
\end{enumerate}
\caption{Bandpass Recovery via Unlimited Sensing }
\label{alg:1}
\end{algorithm}

\subsection{Undersampling with Generalized Modulo Architecture}
\label{sect:genmodulo_recovery}
Here, we assume that $g_{\tbp}\in\PW{\sqb{\ol,\ou}}$ is sampled with $\MOh$ and $\ts>0$. The goal is to recover $g_{\tbp}\rob{k\ts}$ up to an integer multiple of $\lambda_h$. The output samples satisfy
\begin{equation}
\label{eq:gen_mod_decomp_BP}
    y_{\tbp}\sqb{k}=\MOh g_{\tbp}\rob{k\ts}=
    \rob{\gamma_{\tbp}+r_{\tbp}}\sqb{k}
    \EQc{eq:gbpg}
    \rob{\gamma+r_{\tbp}}\sqb{k},
\end{equation}
where $g\rob{k\ts}$ is defined in \eqref{eq:LPF} and the residue $r_{\tbp}\sqb{k}=\sum\nolimits_{p\in\Z} s_p \varepsilon_0\rob{k\ts-\tau_p}$. We assume that the transient is negligible $(\alpha = 0)$, and thus $\varepsilon_0\rob{k\ts-\tau_p}=2\lambda_h \ind_{\left[\tau_p,\infty\right)}\rob{k\ts} \in \cb{0,2\lambda_h}$. Therefore, samples ${\gamma}\sqb{k}$ \eqref{eq:gen_mod_decomp_BP} can be recovered up to $2\lambda_h \mathbb{Z}$ using a variation of \usalg, as explained next. Using the same reasoning as in \usalg, $\forall k \in \mathbb{Z}$ we require that $\vrb{\Delta^N\rob{ \gamma + \eta}\sqb{k}}<\lambda_h$, which is true if $\rob{\ts\Omega e}^N \norm{g}_\infty <\lambda_h$. %
Subsequently, $\widetilde{\gamma}\sqb{k}={\gamma}\sqb{k}+2 M \lambda_h, M\in\Z$, is computed using $\Delta^N \gamma \sqb{k}=\mathscr{M}_{\lambda_h} \rob{\Delta^N y\sqb{k}}$. In the following, we give a theorem guaranteeing recovery from generalized modulo data.
\begin{theorem}[Generalized Modulo Recovery]
\label{thm:GMBP}
Let $g_{\tbp}\in\PW{\sqb{\ol,\ou}}$ and $y_{\tbp}\sqb{k}=\MOh g_{\tbp}\rob{k\ts}$ be the generalized modulo samples of $g_{\tbp}$ with sampling period $\ts>0$. We assume that there exists  $\cndp\in\mathbb{N}$ such that $\ts$ satisfies \eqref{eq:th_USBP}.
Then $g_{\tbp}$ can be recovered from $\Delta^N y_{\tbp}\sqb{k}$ up to an integer multiple of $2\lambda_h$, where $N$ is defined as
\begin{equation}
    N=\left\lceil \frac{\log \lambda_h - \log \beta_g}{\log \rob{\ts\oup e}} \right\rceil,
    \label{eq:N}
\end{equation}
and $\oup$ is defined in \eqref{eq:USLM}.
\end{theorem}
\begin{proof}
From \eqref{eq:gen_mod_decomp_BP} we get that $\Delta^N y_{\tbp}=\Delta^N \rob{\gamma  + r_{\tbp}}$, where $\gamma\sqb{k}=g\rob{k\ts}$ and $g$ is defined in \eqref{eq:LPF}. It can be directly verified that the value of $N$ in \eqref{eq:N} guarantees $\rob{\ts\Omega e}^N \norm{g}_\infty <\lambda_h$, which was also shown in \cite{Bhandari:2020:Ja} for the case $h=0$. Similarly, given that $r_{\tbp}\sqb{k}\in \cb{2 m \lambda_h \mid m\in \Z}$, it can be directly shown that $\Delta^N r_{\tbp}\sqb{k}\in \cb{2 m \lambda_h \mid m\in \Z}$. Then, using a similar reasoning as for \usalg, we get that $\mathscr{M}_{\lambda_h}\rob{\Delta^N r_{\tbp}\sqb{k}}=0$ and thus $\Delta^N\gamma\sqb{k}=\mathscr{M}_{\lambda_h}\rob{\Delta^N y_{\tbp}\sqb{k}}$. Further on, the \usalg can be used to recover $g\rob{t}$. Even though the recovery of $g\rob{t}$ has no direct link to the bandpass property of $g_{\tbp}\rob{t}$, this is required in order to get back $g_{\tbp}\rob{t}$ from $g\rob{t}$. As shown in Theorem \ref{theo:USBP}, this is possible if \eqref{eq:th_USBP} is satisfied.
\end{proof}

\subsection{Fourier-Domain Approach for Modulo Undersampling}
The recovery procedure from non-ideal modulo samples exploits both the parametric structure of the residual signal \eqref{eq:res}, and the partitioning of the modulo signal in the frequency domain. For bandpass signals, two partitioning intervals can be leveraged for the signal recovery, since $\ftsub{g}{\tbp}{\omega}=0$, $\forall \omega \in \R, \abs{\omega} \notin (\ol,\ou)$. As a first observation, we note that the bandpass prior interval of $\omega \in (-\ol,\ol)$ can be used to recover the residual signal's parameters without the need for demodulation, for any $\os > 2\ou$. Qualitatively, this can be observed in \fig{fig:MDP}. However, for further reduction of the sampling frequency, demodulation will be necessary. 

\bpara{Fourier-Domain Partitioning of Modulo Bandpass Signals.}
\label{subsec:CDM2} 
We use the DFT to develop the modulo bandpass sampling conditions in the discrete frequency domain. The function $g_{\tbp}\rob{t}$ is considered to be $\tau$-periodic. To avoid spectral leakage, we assume that $\ts=\frac{\tau}{K} $, where $K\in\Z$ denotes the number of samples. Moreover, without reducing the generality, we assume that $\ol$ $\ou$ are integer multiples of ${\frac{2\pi}{\tau}}$. Then we define $\ql \in \Z$ and $\qu  \in \Z$ as
\begin{equation}
\label{eq:QLQU}
    \ql \DE {\frac{K \ol}{\os}}, \quad \qu \DE {\frac{K \ou}{\os}},
\end{equation} 
For $\ts > 0$, sampling $y_{\tbp}\rob{t}$ over a period produces $K = \tau/\ts$ discrete samples. In terms of the $K$--length DFT of $y_{\tbp}\sqb{k}$, \eqref{eq:FSBP2} becomes,
\begin{align*}
     \dftsub{y}{\tbp}{n} &=  \ts \sum\limits_{k=0}^{K-1} {y_\tbp \sqb{k}{e^{ - \jmath \frac{2\pi}{K} k n}}} = \ftsub{y}{\tbp \delta}{\omega} \vert_{\omega = \tfrac{2\pi}{\tau}n} \notag \\
     &= 
     \left( {{{\hat \gamma }_\delta } - {{\widehat r}_{\mathsf{BP}\delta }}} \right){\left. {\left( \omega  \right)} \right|_{\omega  = \frac{{2\pi }}{\tau }n}} \notag \\ 
     &= \dft{g}{n} - \dftsub{r}{\tbp}{n}.
     \end{align*}
Moreover, $g\rob{t}$ defined in \eqref{eq:LPF} is $\tau$-periodic as well, given that its frequencies are obtained by shifting $\ftsub{g}{\tbp}{\omega}$ by integer multiples of $\os=\frac{2\pi}{\ts}$.  

Given that we assume $\ol >0$, we get $\ol \geq 2\pi / \tau$ and thus $\ql \geq 1$. 
We know that the values of the DFT of $\gamma[k]=g(k\ts)$ map to the magnitudes of $\widehat{g}(\omega)\vert_{\omega=n\tfrac{2\pi}{\tau}=n\tfrac{\os}{K}}$.
It then follows from \eqref{eq:USLM} and \eqref{eq:QLQU} that the indices of the passband frequencies in $\dft{g}{n}$, the DFT of $g(t)$, are contained in $$\cb{-\qup,\qup+1,\dots,-\qlp}\cup\cb{\qlp,\qlp+1,\dots,\qup},$$ where
\begin{align}
\label{eq:FPLM}
&    \qup = \begin{cases}
       \qu - \frac{\cndp-1}{2}K  & \qquad\cndp \in 2\N+1 \\
        \frac{\cndp}{2}K -\ql  &\qquad \cndp \in 2\N
    \end{cases} \mbox{ and},  \\
&    \qlp = \qup - \qlb, \notag
\end{align}
where $\qlb \DE \qu - \ql$  and $\cndp-1=\frac{2\ol}{\os} \in \N$ represents the number of periods $\os$ within $\rob{-\ol , \ol}$.
We have that $g(k\ts)=g_{\tbp}(k\ts)$ via \eqref{eq:gbpg}.
Then, using the periodicity of the DFT $\dftsub{y}{\tbp}{K-q}=\dftsub{y}{\tbp}{q}, \forall q\in\cb{0,1,\dots,K}$, the following holds 
\begin{equation*}
        \dftsub{y}{\tbp}{n} = \begin{cases}
           \dft{g}{n} - \dftsub{r}{\tbp}{n}, \quad &n \in \eset{\qlp}{\qup}{K} \\
           -\dftsub{r}{\tbp}{n}, \quad &n \in \mathbb{I}_{K} \setminus \eset{\qlp}{\qup}{K}
       \end{cases}  
\end{equation*}
where the set $\eset{\qlp}{\qup}{K}$ is given by, $$\eset{\qlp}{\qup}{K} = \left[ {\qlp,\qup} \right] \cup \left[ { K - \qup, K - \qlp} \right].$$
The partition is schematically illustrated in \fig{fig:FDpartition}. As can be seen, the two sets on which the residual and the lowpass signals can be separated are,
\begin{equation}
\mathbb{F}_{\tbp}=\begin{cases} 
\sqb{\qup+1, K - \qup -1} & \mbox{(Outer Set)} \\
\sqb{0, \qlp-1} \cup \sqb{K - \qlp -1, K-1} & \mbox{(Inner Set)}
\end{cases}. 
\label{eq:inner_outer_set}
\end{equation}
A sampling density criterion that enables recovery based on the two sets is given in the next theorem.

 \begin{figure}[!t]
\centering
\includegraphics[width =0.65\textwidth, angle=0]{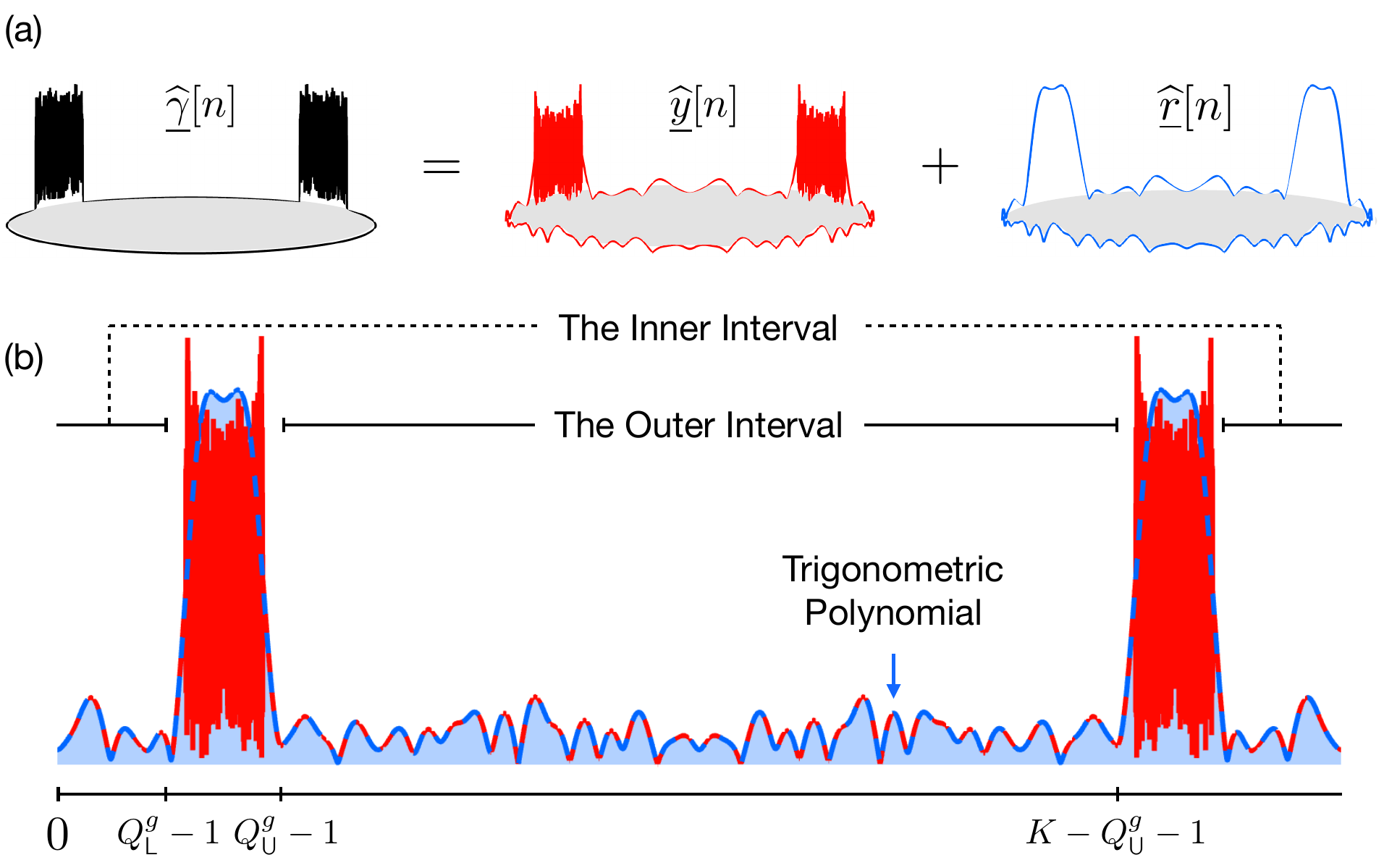}
\caption[width =1\textwidth]{The partitioning of a bandpass signal into its modulo (red) and residual (blue) components in the discrete Fourier-domain. (a) Decomposition of the bandpass spectra into modulo and residual components. (b) The inner and outer intervals as defined in \eqref{eq:inner_outer_set} on a single DFT period.}
\label{fig:FDpartition}
\end{figure}

\begin{theorem}[Bandpass Fourier-Domain Recovery]
\label{thm:FPBP}
Let $ \BP{g_{\tbp}}$ be a $\tau$-periodic function. Suppose that we are given $K$ modulo samples of $y_{\tbp}\sqb{k} = \MONI{g_{\tbp}\rob{k\ts}}$, at sampling period $\ts=\frac{\tau}{K} >0$, folded at most $M$ times. If there exists $\cndp \in \N$ such that $\ts$ satisfies,
\begin{align}
&\begin{cases}
    \frac{\rob{\cndp-1}\tau}{2\ql} \leq \ts \leq \frac{\cndp \tau}{2\rob{\qu + M + 1}}, & \cndp \in 2\N+1 \\
    \frac{(\cndp-1) \tau}{2\rob{\ql - M - 1}} \leq \ts \leq \frac{\cndp\tau}{2\qu}, & \cndp \in 2\N
\end{cases} \notag \\
&\mbox{\textit{or,}} \notag \\
&\begin{cases}
    \frac{\rob{\cndp-1}\tau}{2\rob{\ql - M-1}} \leq \ts \leq \frac{\cndp \tau}{2\qu}, & \cndp \in 2\N+1 \\
    \frac{\rob{\cndp-1}\tau}{2\ql}  \leq \ts \leq \frac{\cndp \tau}{2\rob{\qu + M + 1}}, & \cndp \in 2\N
\end{cases},
\label{eq:outer}
\end{align}
and $\ql > M+1$ for $\cndp \in 2\N$ or $\cndp \in (2\N+1) \setminus \{1\}$ in \eqref{eq:outer}. 
Then, $g_{\tbp}\rob{t}$ can be recovered from $y_{\tbp}\sqb{k}$ up to $2\lambda\mathbb{Z}$.

\label{theo:FDBP}
\end{theorem}

\begin{proof}
We start from \eqref{eq:bps_cond} and multiply by $K$ to get,
\begin{equation}
\label{eq:DBPS0}
    \rob{\frac{\cndp-1}{2}}K \leq \frac{K\ol}{\os} \leq \frac{K\ou}{\os} \leq \rob{\frac{\cndp}{2}}K.
\end{equation}
Re-arranging \eqref{eq:DBPS0} with respect to the number of samples $K$,
\begin{equation}
\label{eq:DBPS1}
    {\frac{2}{\cndp}}\frac{K\ou}{\os} \leq K \leq {\frac{2}{\cndp-1}}\frac{K\ol}{\os}.
\end{equation}
Then, using \eqref{eq:QLQU}, we get an equivalent of the bandpass sampling condition \eqref{eq:bps_cond} in the discrete frequency domain
\begin{equation}
\label{eq:DBPS}
    {\frac{2 \qu}{\cndp}} \leq K \leq {\frac{2 \ql}{\cndp-1}}.
\end{equation}

Next, to recover the bandpass samples we note that for a modulo sequence with $M$ folds, the first order difference of the residual, $\bar{r}\sqb{k}$, is a $2M$-parametric signal. Its recovery requires $2M+1$ samples in the outer set \eqref{eq:inner_outer_set}, for which $\dftsub{y}{\tbp}{n} = -\dftsub{r}{\tbp}{n}$. This is guaranteed by $K -2\qup-1 \geq 2M+1$.
By substituting with $\qup$ from \eqref{eq:FPLM} and isolating $K$ we get,
\begin{equation*}
\begin{cases} 
    K \geq \frac{2\rob{\qu + M + 1}}{\cndp} & \cndp \in 2\N+1 \notag \\
    K \leq \frac{2\rob{\ql-M-1}}{\cndp-1} & \cndp \in 2\N
\end{cases}.    
\end{equation*}
For every $M\geq0, \cndp >0$, we get that, 
$$\frac{2\rob{\qu + M + 1}}{\cndp} \geq \frac{2\qu}{\cndp} \mbox{ and } \frac{2\rob{\ql-M-1}}{\cndp-1} \leq \frac{2\ql}{\cndp-1}.$$ Integration with the bound in \eqref{eq:DBPS} and substitution $K = \tau / \ts$ yields the first condition in \eqref{eq:outer}.
For the set  $$ \sqb{0, \qlp-1} \cup \sqb{K - \qlp -2, K-2}$$ we require that $2\rob{\qlp-1} \geq 2M$. We use $\qlp$ from \eqref{eq:FPLM} and isolate $K$ to get,
\begin{equation}
\begin{cases} 
    K \leq \frac{2\rob{\ql-M-1}}{\cndp-1} & \cndp \in 2\N+1 \\
    K \geq \frac{2\rob{\qu + M + 1}}{\cndp} & \cndp \in 2\N
\end{cases}.    
\label{eq:inner0}
\end{equation}
As in the outer set case, we integrate \eqref{eq:inner0} with the bound in  \eqref{eq:DBPS}, and replace $K = \tau / \ts$ to get the sampling condition in \eqref{eq:outer}. For $\ts$ that satisfies \eqref{eq:outer} for some $\cndp \in \N$, $g_{\tbp}\rob{t}$ can be recovered from its modulo samples by Algorithm \ref{alg:2}.
\end{proof}
\noindent Remarks on Theorem~\ref{theo:FDBP}:
\begin{enumerate}[leftmargin =45pt,label=\ding{224}]
    \item For $\qu =\qlb$ (bandlimited signal case), we return to the expected result that $2(M+1)$ modulo samples beyond the maximal non-zero frequency of the signal are needed.
    \item For $\cndp=1$, reconstruction from the inner set samples requires $\ql \geq M+1$, regardless of the sampling period.
    \item As $\cndp$ gets larger, more replicas of the bandwidth occupy the outer set (even $\cndp$ case) or the inner set (odd $\cndp$ case). Consequentially, fewer folds can be recovered.
    \item For a given $\cndp$, the outer set (even $\cndp$ case) or the inner set (odd $\cndp$ case) shrinks as $K$ increases. Consequentially, larger $M$ would require lowering the sampling period for recovery based on a specific interval. This implies that in modulo undersampling, the sampling period should be optimized based on the recovery strategy (the exploited interval, and the parity of $\cndp$). 
\end{enumerate}
\begin{algorithm}[t!]
\SetAlgoLined
{\bf Input:} $\{y_{\tbp}\sqb{k}\}_{k=0}^{K-1}$, $\tau,\ol, \ou$, and $M = |\mcal{M}|$ \eqref{eq:res}\\ 
\KwResult{Bandpass function, $\widetilde g_{\tbp}\rob{t}$. }

\begin{enumerate}[label = $\arabic*)$,leftmargin=30pt,itemsep=1pt]
\item Compute $\dftsub{\bar{y}}{\tbp}{n}$, the DFT of $\bar y_{\tbp}\sqb{k}$.
\item Compute $\qup$ or $\qlp$ using \eqref{eq:FPLM}.
\item Define $ \dft{\bar{r}}{n} \DE -\dftsub{\bar{y}}{\tbp}{n} \equalref{\eqref{eq:yBPtoy}} -\dft{\bar{y}}{n}$, where $n \in \mathbb{F}_{\tbp}$ \eqref{eq:inner_outer_set}. \\
\item Use spectral estimation to estimate $\widetilde {\bar r}\sqb{k}$ \cite{Bhandari:2021:J}. 
\item Estimate $\widetilde r\sqb{k}=\sum_{m=1}^k \widetilde {\bar r}\sqb{m}$ (up to an unknown constant), and compute $\widetilde \gamma\sqb{k}=\widetilde {r}\sqb{k} + y_{\tbp}\sqb{k}$.
\item Compute $\ftsub{\gamma}{\delta}{\omega} = \ts \sum_{k \in \Z} \gamma\sqb{k} e^{-\jmath \omega k \ts}$.
\item Estimate $\widetilde g_{\tbp}\rob{t}$ from $\ftsub{\gamma}{\delta}{\omega}$ using \eqref{eq:Fshift} and $\cndp$ from \eqref{eq:outer}.
\end{enumerate}
\caption{Bandpass Recovery in the Fourier-Domain }
\label{alg:2}
\end{algorithm}
Similarly to Corollary \ref{cor:AMUS}, the prevalence of AM signals and their spectral structure motivates the proof of the next corollary for non-ideal modulo bandpass sampling. 
\begin{cor}[Fourier-Domain Recovery for AM Signals]
\label{cor:AMFP}
If $g_{\tbp}(t) \BP{}$ is a $\tau$-periodic function such that the side lobes of $\hat{g}$ are symmetric, i.e.,  $\hat{g}\rob{\ol+\omega}=\hat{g}\rob{\ou-\omega},  \forall\omega\in\sqb{0,\blb/2}$, then $g_{\tbp}(t)$ can be recovered from $K$ modulo samples folded at most $M$ times via Algorithm \ref{alg:2} if there exists $\cndpAM\in \N$ such that $\os =\frac{2\pi}{K \ts} = \frac{\ou+\ol}{2\cndpAM}$  and
\begin{equation}
    \blb \leq \os \leq \frac{4\pi\rob{\ql-M-1}}{\tau(2\cndpAM-1)}. 
    \label{eq:AMK}
\end{equation}
\end{cor}
\begin{proof}
Just as in Corollary \ref{cor:AMUS}, $\os= \rob{\ou+\ol} / \rob{2\cndpAM}$ guarantees perfect overlapping of the spectral lobes. For recovery to work, we further require that $\blb \leq \os$. Here, only the outer set $\sqb{\qup+1, K - \qup -2}$ can be used and hence, we require that $\rob{K -2\qup-2} \geq 2M$. Then, the RHS of \eqref{eq:AMK} follows by substituting $\qup$ \eqref{eq:FPLM} (with $\cndp=2\cndpAM$) and $K=\rob{\os \tau}/(2\pi)$.
\end{proof}
\section{Experiments}
\begin{figure}[!t]
    \centering
    \includegraphics[width=0.75\textwidth]{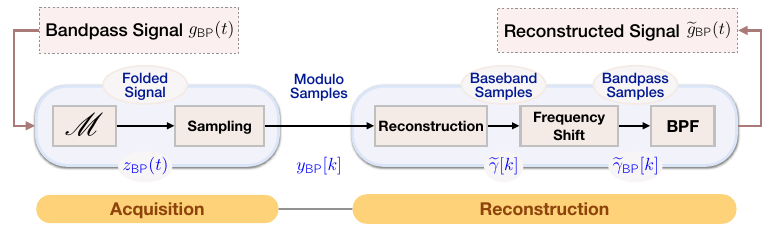}
    \caption{Modulo bandpass acquisition and reconstruction pipeline.}
    \label{fig:pipeline}
\end{figure}

This section provides experiments that validate the theoretical results and demonstrate the potential of our modulo undersampling technique based on the sampling periods in Section \ref{sec:BPwithOL}. We demonstrate computational demodulation and recovery of a bandpass signal from its modulo samples.
\subsection{Numerical Simulations}
As the most practical use case, we consider the case of baseband relocation. 
In the experiments, the bandwidth of the positive Fourier-domain coefficients is $\ou - \ol = \pi$, and their complex values are generated randomly so $\ftsub{g}{\tbp}{w} = 100\texttt{U}_0+ \jmath 120\texttt{U}_1$, $\forall \omega \in [\ol, \ou]$, where $(\texttt{U}_0,\texttt{U}_1)\in (0,1)\times(0,1)$ are sampled from a uniform distribution. The negative Fourier-domain coefficients are the complex conjugate of the positive Fourier-domain coefficients. The bandpass signal is normalized so $\maxn{g_{\tbp}}=1$. The experimental pipeline follows the block diagram in \fig{fig:pipeline}. The modulo of the bandpass signal, $z_\tbp\rob{t}$ is sampled at period $\ts$ for a time interval $\tau$ to produce a sequence of $K=\tau/T$ modulo samples $y_{\tbp}\sqb{k}$, and the samples that are recovered by Algorithm \ref{alg:1} and Algorithm \ref{alg:2} are denoted by  $\widetilde{\gamma}_{\tbp}^{\mathsf{US}}\sqb{k}$ and $\widetilde{\gamma}_{\tbp}^{\mathsf{FP}}\sqb{k}$, respectively. Both algorithms recover the signal up to an additive constant of $2\lambda\Z$. This factor is estimated from the ground-truth samples $\gamma_{\tbp}\sqb{k}$, and the MSE between the signals is measured and denoted by $\mse{\gamma_{\tbp}}{ \widetilde{\gamma}_{\tbp}^{\mathsf{US}}}$ and  $\mse{\gamma_{\tbp}}{ \widetilde{\gamma}_{\tbp}^{\mathsf{FP}}}$.

\bpara{---$\bullet$ Time-Domain Recovery.} The evaluation on this experiment is based on $1000$ random realizations of a bandpass signal $g_{\tbp}\in \mathcal{B}_{\left(50\pi, 51\pi \right)}$ with $\lambda\sim\texttt{U}\rob{0.05, 0.1}$. According to \eqref{eq:TUS}, the sampling period should satisfy $\ts \leq \TUS^{\tbp}= 1  /\rob{102\pi e} \approx 1.1\times10^{-3}$. However, according to Theorem \ref{theo:USBP} with $\cndp=5$ recovery is still possible for $0.080 \leq \ts \leq 0.085$. 
We set $\ts = 0.080$ to achieve baseband relocation in \eqref{eq:basebandrates}. The maximal reconstruction error achieved by Algorithm \ref{alg:1} is $\mse{\gamma_{\tbp}}{ \widetilde{\gamma}_{\tbp}^{\mathsf{US}}}= 8.54\times 10^{-32}$. This verifies the validity of theorem \ref{theo:USBP}, and shows a factor of $8\pi e$ reduction in the required sampling period for bandpass signals recovery. The modulo samples of $g_{\tbp}\rob{t}$ with $\lambda= 0.07$ and the reconstruction using Algorithm \ref{alg:1} with $N=3$ are shown in  \fig{fig:basbandsubsamp}. 

 \begin{figure}[!t]
\centering
\includegraphics[width =0.65\textwidth]{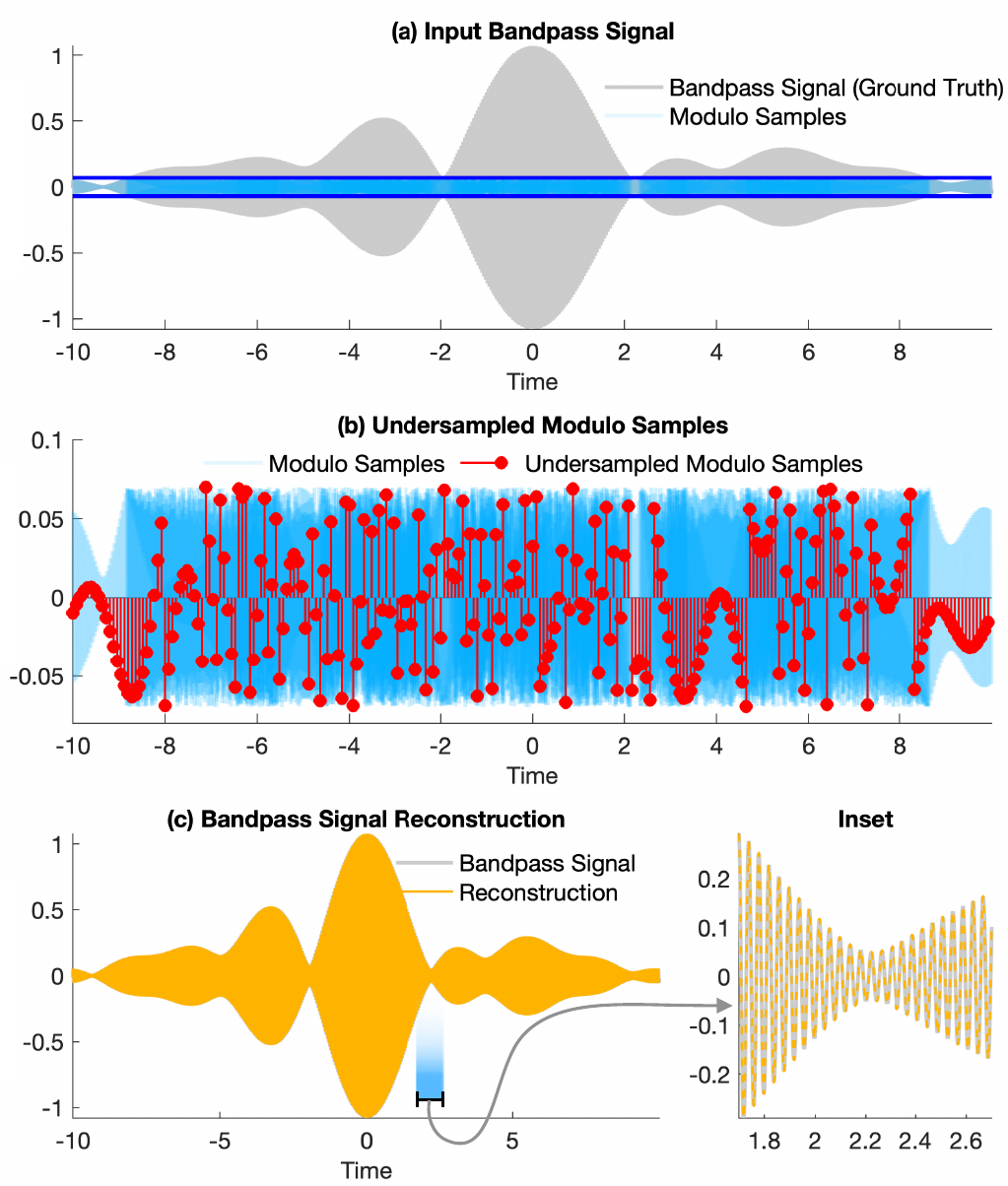}
\caption{Reconstruction of a bandpass signal with $\ol=50\pi, \ou=51\pi$ from its modulo samples at $\ts = 8\times10^{-2}$, where $\lambda=0.07$, $\cndp=5$. Reconstruction with Algorithm \ref{alg:1} achieves $\mse{\gamma_{\tbp}}{ \widetilde{\gamma}_{\tbp}^{\mathsf{US}}} = 2.06\times 10^{-34}$.}
\label{fig:basbandsubsamp}
\end{figure}

\bpara{---$\bullet$ Fourier-Domain Recovery.} To test the utility of our numerical algorithm, $1000$ random realizations of the signal $g_{\tbp}\in \mathcal{B}_{\left(199\pi, 200\pi \right)}$ and $\lambda \sim \texttt{U}\rob{0.02, 0.1}$ were generated. The maximum number of detected folds was $M=259$. Without undersampling, recovery of $g_{\tbp}$ using the \usfp algorithm requires a sampling period smaller than $4.4\times10^{-3}$. 
Based on \eqref{eq:outer}, we observe that with $\cndp=2$ the sampling period can be eased to $5.8\times10^{-3} \leq \ts \leq 10\times10^{-3}$.
We set $\ts = 10\times10^{-3}$, which relocates the bandpass signal to a baseband position \eqref{eq:basebandrates}. The maximal reconstruction error using Algorithm \ref{alg:2} is $\mse{\gamma_{\tbp}}{ \widetilde{\gamma}_{\tbp}^{\mathsf{FP}}} = 6.39\times 10^{-32}$.

\subsection{Hardware Experiments}
\bparab{Experimental Protocol.}
The results for non-ideal modulo architectures can be validated via hardware experiments, by achieving computational demodulation and recovery of a bandpass signal from its modulo samples.
The input signal follows,  
\begin{equation}
    g_{\tbp}\rob{t} =A_\mathsf{M}\left(1+\cos{\left( \Omega_\mathsf{M} t + \theta_\mathsf{M} \right) }\right) \sin{\left(\Omega_\mathsf{C} t+\theta_\mathsf{C}\right)}
    \label{eq:AM}
\end{equation} 
where $\left\{ A_\mathsf{M}, \Omega_\mathsf{M}, \theta_\mathsf{M} \right\}$ are the amplitude, frequency and phase of a message wave, and $\left\{\Omega_\mathsf{C}, \theta_\mathsf{C}\right\}$ are the frequency and phase of a carrier wave. The minimal and maximal frequencies of the bandpass signal in \eqref{eq:AM} are $\ol = \Omega_\mathsf{C} - \Omega_\mathsf{M}$ and $\ou = \Omega_\mathsf{C} + \Omega_\mathsf{M}$, respectively. In the experimental prototype, the signal \eqref{eq:AM} is generated by a \texttt{TTI--TG5011} function generator, and then fed into the \usadc and a $4$-channel \texttt{DSO--X 3024A} oscilloscope that simultaneously samples the input and the output of the \usadc to produce the ground-truth and the modulo samples. The sampling and reconstruction pipeline follows \fig{fig:pipeline}. For validation purposes, standard curve fitting techniques are used to estimate the parameters of \eqref{eq:AM}.

\spara{Exp.~1: Reconstruction from Non-ideal Modulo Samples.} 
To show demodulation and reconstruction using Algorithm \ref{alg:2}, a signal with a period of $\tau=299\ \mathrm{ms}$ that obeys \eqref{eq:AM} with the parameters $A_\mathsf{M}=3.72, \Omega_\mathsf{M}=43.98\ \mathrm{rad/s}, \theta_\mathsf{M}=2.64, \Omega_\mathsf{C}=628.31\times10^3\ \mathrm{rad/s}$, and $\theta_\mathsf{C}=-0.93\times10^{-2}\ \mathrm{rad}$ is generated. The signal was fed into the \usadc with a threshold $\lambda \approx 2.01 \pm 3/20$, where $\pm 3/20$ is a manually adjustable design parameter, and then sampled at period $\ts = 10\ \mathrm{ms}$. The produced modulo samples $y\sqb{k}$ have $\qup=3$ and $M=4$. As \eqref{eq:AM} is an AM signal, we validate that $\os = 2\pi / \ts = 200\pi\ \mathrm{rad/s}$ satisfies \eqref{eq:AMrates} with $\cndpAM=1$ and also satisfies \eqref{eq:AMK} since $\blb = 28\pi \leq 200\pi \leq 293\pi$. The conditions of Corollary \ref{cor:AMFP} are satisfied so we can use Algorithm \ref{alg:2} with input $y\sqb{k}$, to produce the recovered samples $\widetilde{\gamma}^{\mathsf{FP}}\sqb{k}$. The error compared to the sequence $\gamma\sqb{k} =g\rob{k \ts}$ is $\mse{\gamma}{\widetilde{\gamma}^{\mathsf{FP}}} = 7.1\times 10^{-3}$, which verifies that $y_\tbp\sqb{k} \approx y\sqb{k}$ when undersampled. We then use interpolation with $\mathrm{sinc}\rob{\os t}$ to get $\widetilde{g}^{\mathsf{FP}}\rob{t}$, and recover $\widetilde{g}_\tbp\rob{t}= \widetilde{g}^{\mathsf{FP}}\rob{t}\tfrac{\sin{\left(\Omega_\mathsf{C} t+\theta_\mathsf{C}\right)}}{\sin\theta_\mathsf{C}}$. The corresponding signals are shown in \fig{fig:exp3}.

\begingroup
\begin{table*}[!t]
\centering 
\renewcommand*{\arraystretch}{1.1}
\caption{Hardware Experiments--Summary of Parameters and Performance}
\label{tab:exp}
\resizebox{\textwidth}{!}{%
\begin{tabular}{|c|cccccccccccccccc|}
\hline\hline
Exp. & Fig.~No. & $\ts$ ($\mathrm{ms}$) & $K$ & $\tau$ ($\mathrm{ms}$) & $\ol (\mathrm{rad/s})$ & $\ou (\mathrm{rad/s})$ & $\rho$ & $A_\mathsf{M}$ & $\Omega_{\mathsf{M}} (\mathrm{rad/s})$ & $\theta_\mathsf{M} (\mathrm{rad})$ & $\Omega_{\mathsf{C}} (\mathrm{rad/s})$ & $\theta_\mathsf{C} (\mathrm{rad})$& $\mse{\gamma}{\widetilde{\gamma}}$ \\ \hline
$0$ & \fig{fig:exp0} &  $2.5$ & $24$ & $59.8$ & $2.29\times 10^{3}$ & $2.73\times 10^{3}$ & $2.49$ & $-2.502$ & $219.91$ & $1.147$ & $2.51\times 10^{3}$ & $-0.17$ & $1.7\times 10^{-3}$  \\
$1$ & \fig{fig:exp3} &  $10$ & $30$ & $299$ & $584.33$ & $672.30$ & $3.77$ & $3.72$ & $43.98$ & $2.644$ & $628.31$ & $-0.93$  &  $7.1\times 10^{-3}$ \\
$2$ & \fig{fig:gen_modulo} & $2.5$ & $24$ & $59.8$ & $2.29\times 10^{3}$ & $2.73\times 10^{3}$ & $2.6$ & $2.47$ & $216$ & $-1.57\times10^{-2}$ & $2.5\times10^{3}$  & $-1.57\times10^{-2}$ & $4.26\times 10^{-4}$  \\
\hline
\end{tabular}%
}
\end{table*}
\endgroup

 \begin{figure}[!t]
\centering
\includegraphics[width =0.6\textwidth]{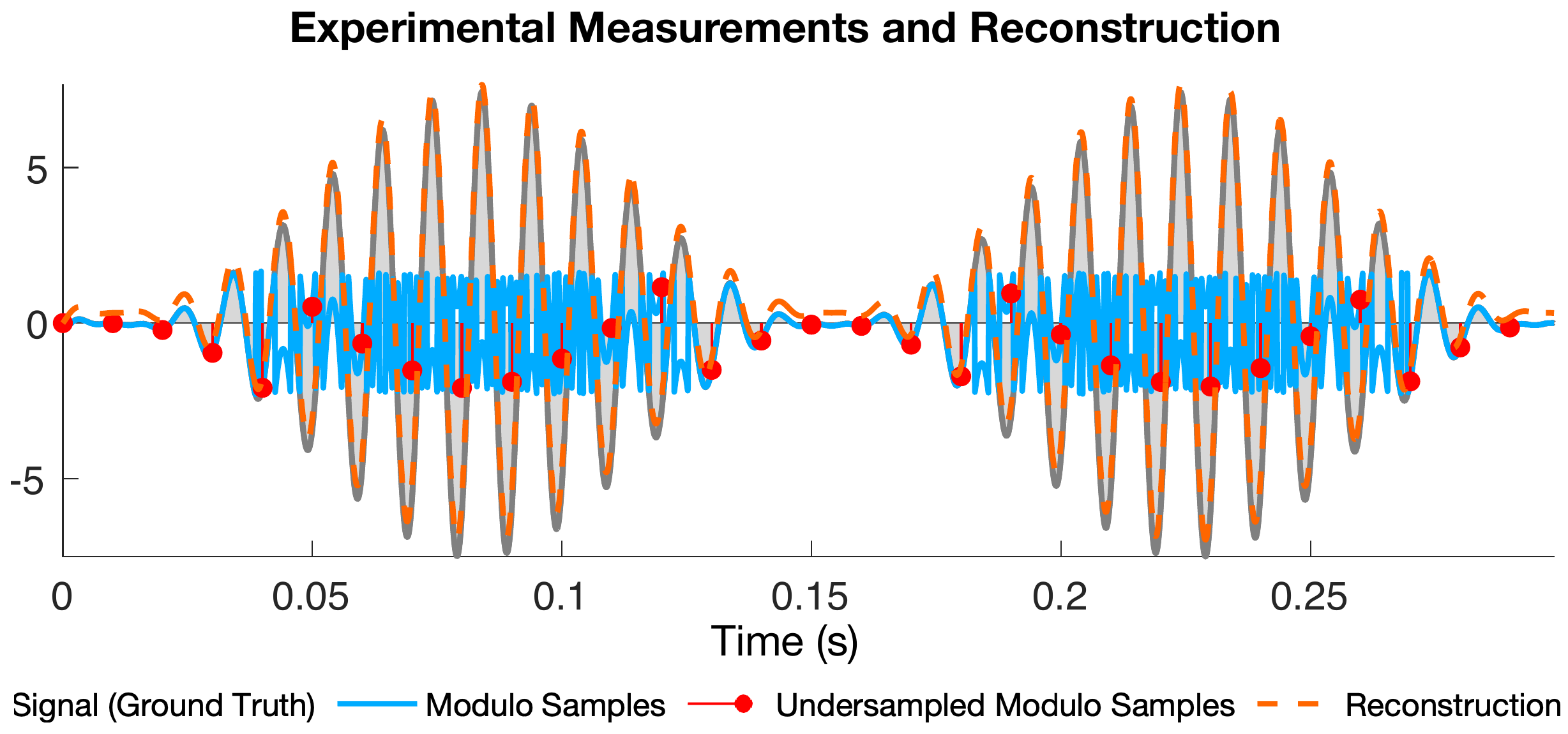}
\caption{Reconstruction from experimental data using \usadc. The experimental parameters are in \rftab{tab:exp} (Exp.~1).}
\label{fig:exp3}
\end{figure}

\begin{figure}[!t]
    \centering
    \includegraphics[width =0.65\textwidth]{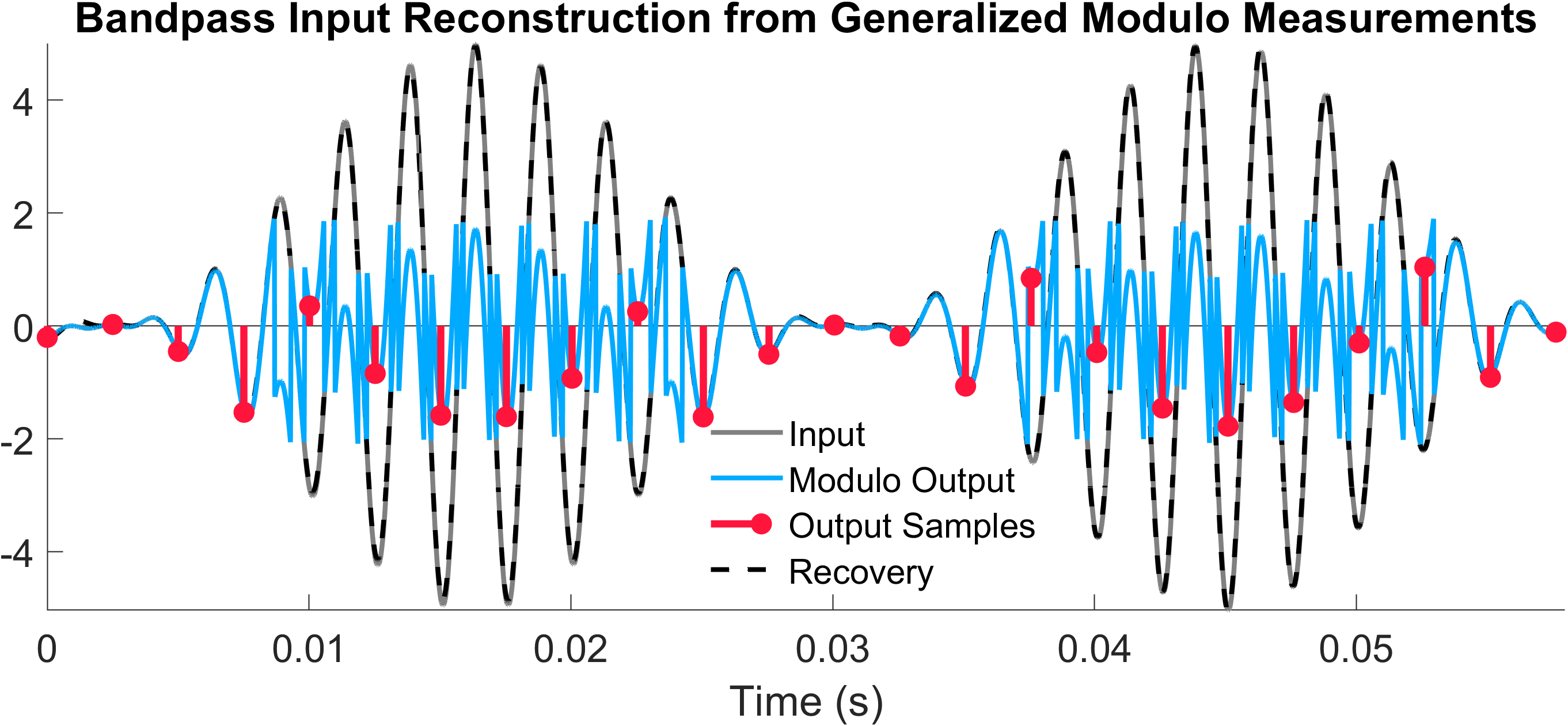}
    \caption{Reconstruction from experimental data using the generalized \madc. The experimental parameters are in \rftab{tab:exp} (Exp.~2).}
    \label{fig:gen_modulo}
\end{figure}

\spara{Exp.~2: Reconstruction From Generalized Modulo Nonlinearity.} The input is generated using \eqref{eq:AM} with  $\Omega_\mathsf{M}=216\ \mathrm{rad/s}, \Omega_\mathsf{C}=2.5\times10^3\ \mathrm{rad/s}, \theta_\mathsf{M}=\theta_\mathsf{C}=-1.57\times10^{-2}\ \mathrm{rad}$, and $A_\mathsf{M}=2.47$. The generalized modulo nonlinearity parameters are $\lambda=1.93, h=0.88$, leading to a ratio $\rho=2.6$ between the input amplitude and modulo threshold. The input is bandpass, satisfying $\{\ol,\ou\} =\{2.29, 2.73\}\times10^{3}$. The input and output of the modulo were sampled with period $\ts=2.5\times10^{-3}$, which is more than double the Nyquist sampling period of $\TNQ=1.2\ \mathrm{ms}$. This resulted in $K=24$ samples. The input $g_{\tbp}\rob{t}$, modulo output $z\rob{t}$ and output samples $y\sqb{k}=z\rob{k\ts}$ are depicted in \fig{fig:gen_modulo}. The samples of the baseband signal $g\rob{k\ts}$ were subsequently recovered using Algorithm \ref{alg:1} with threshold $\lambda_h$. The signal $g\rob{t}$ is reconstructed via interpolation with $\mathrm{sinc}\rob{\Omega t}$, where $\Omega=\tfrac{\pi}{\ts}$. The input $g_{\tbp}\rob{t}$, recovered as $\widetilde{g}_{\tbp}\rob{t}=\widetilde{g}\rob{t}\tfrac{\sin{\left(\Omega_\mathsf{C} t+\theta_\mathsf{C}\right)}}{\sin\theta_\mathsf{C}}$, is depicted in \fig{fig:gen_modulo}. The error between the recovered input and ground truth is $\mse{\gamma_{\tbp}}{\widetilde{\gamma}_{\tbp}}=4.26\times10^{-4}$.

\section{Conclusions and Future Work}
In this paper, we proposed a computational sensing-based solution for the efficient acquisition of bandpass signals with a high dynamic range. 
Our approach is based on modulo folding architectures, whose distinct advantages had been previously proven for bandlimited signal types, but such advantages have not yet been capitalized for bandpass signals. 
In the sequel, we have shown that the characteristics of a bandpass input signal can be leveraged to allow a unique mapping between bandpass and modulo samples at sub-Nyquist rates. Based on that, we provided new sampling conditions guaranteeing recovery in the case bandpass and periodic bandpass signals. Considering the practicalities of modulo \adcs, we analyzed three different folding architectures, namely, the ideal modulo, non-ideal modulo and generalized modulo. Furthermore, feasible recovery algorithms that show that recovery is possible at the reported rates in both time- and frequency-domains are presented as well. Our numerical experiments validated the underlying theoretical findings. Taking a step closer to practice, we validated our algorithms with hardware experiments, thus establishing their practical advantages.

The main message of this work is that folding architectures are suitable for undersampling-based acquisition. Several real-life signal processing systems implement undersampling using \adcs but suffer from dynamic range limitations, thus creating a research gap; our work on the other hand shows its relevancy to a wide range of real-life signal processing systems.  

\ipara{Future Work:} The next step forward will be a study of the framework's integration into different architectures that leverage bandpass sampling in hardware devices \cite{Fazi:1997:C, Kappes:2003:J,Seo:2003:J}, and the challenges that arise in the context. This may include

\begin{enumerate}[leftmargin = 30pt, label = ---$\bullet$]
\item  Multi-channel time-interleaved \adc, where the samples follow a non-uniform, periodic pattern and aliasing is present in the baseband signal as well \cite{Wahab:2022:J}.

\item Sigma-Delta \adcs, where a feedback loop is utilized for noise shaping around a centered frequency \cite{Kappes:2003:J}.

\item Software-defined radios, where algorithms are required to effectively filter and process multiple passbands.

\end{enumerate}

\section*{Appendix}

\begin{table}[!h]
\centering
\caption{Frequently Used Symbols}
\rowcolors{2}{blue!4}{white}
\renewcommand{\arraystretch}{1}
\begin{supertabular}{m{2cm}	 || m{10cm} l}
	Symbol & Definition \\ \hline  \hline 
	$\mathbbm{1}_{\mathcal{D}} \rob{t}$ & Indicator function on the domain $\mathcal{D}$. \\
	$\mathscr{M}_\lambda$ & Ideal modulo non-linearity with threshold $\lambda$.\\
	$\widetilde{\mathscr{M}}_\lambda$ & Non-ideal modulo non-linearity with threshold $\lambda$.\\
	$\MOh$ & Generalized modulo operator with $\boldsymbol{\mathsf{H}}=[\lambda,h,\alpha]$ where $h$ is the hysteresis parameter and $\alpha$ is the transient duration.\\	
	$g_{\tbp}(t),g(t)$ & Bandpass and corresponding bandlimited input signals.\\
	$\os, \ts$ & Sampling frequency $(\mathrm{rad/s})$ and period $\ts = 2\pi / \os$ ($\mathrm{s}$).\\
	$\gamma[k], \gamma_{\tbp}[k]$ & Input samples $\gamma_{\tbp}[k]=g_{\tbp}(k\ts), \gamma[k]=g(k\ts)$.\\
	$y[k], y_{\tbp}[k]$ & Modulo samples of lowpass and bandpass functions. \\
	$r[k], r_{\tbp}[k]$ & Residual samples of lowpass and bandpass functions. \\
	$\bar{f} $ & First order difference, $\bar{f}\sqb{k}=\Delta f = f\sqb{k+1}-f\sqb{k}$.\\
	$\Delta^N(\cdot)$ & Finite difference operator of order $N$.\\
	$\ft{f}{\omega}$     &    Continuous-time Fourier transform (FT) of $f\rob{t}$.\\
	$\ftsub{f}{\delta}{\omega}$ & Normalized Fourier series (FS) with coefficients $f[k]$. \\
	$\dft{f}{n}$    &     Discrete Fourier transform (DFT) of the sequence $f\sqb{k}$. \\
	$\ol, \ou$ & Minimum and maximum frequencies of a function. \\
	$\blb$ & Passband bandwidth $\blb = \ou - \ol$. \\
	$\oup$ & Bandwidth of baseband function $g$ computed via \eqref{eq:USLM}. \\
	$\ql, \qu, \qup$ & Indices for $\ol,\ou,\oup$ on a discrete, normalized  axis \eqref{eq:QLQU}. \\
	$\mathcal{B}_{\Omega}$    &     Space of $\Omega$-bandlimited functions.\\
	$\mathcal{B}_{\left(\ol, \ou \right)}$ \vspace{8pt} & Space of functions with spectrum inside $[\ol,\ou]$\vspace{4pt}.\\
	\rowcolor{white}
	$\mse{\mathbf{x}}{\mathbf{y}}$ & Mean squared error between vectors $\mat{x}$  and $\mat{y}$.	\\
\end{supertabular}
\label{tab:my-table}
\end{table}

\bpara{Notation.} The sets of real, integer and natural numbers are denoted by $\R$, $\Z$, $\N$ respectively. We denote by $f\rob{t}$, $t \in \R$ continuous-time functions and by $f\sqb{k}$, $k \in \Z$ discrete-time sequences. The $p$-norm for functions in $f\in\LL{p}$ is denoted as $\nrm{f}{p}$, for $1\leq p<\infty$. The Fourier series (FS) representation of a $\tau$-periodic $f\rob{t}$ is,
\begin{equation*}
    f\rob{t} \DE \sum\limits_{k \in \Z} f_k e^{-\jmath \tfrac{2\pi}{\tau} t k}, \quad f_k = \tfrac{2\pi}{\tau}\int_{\tau/2}^{\tau/2}{f\rob{t}}e^{\jmath \tfrac{2\pi}{\tau} t k}dt.
\end{equation*}
The notation $\ft{f}{\omega}$ is used for the Fourier transform of a continuous-time function $f\rob{t} \in \LL{1}$. The sampled or discrete Fourier transform (DFT) of sequence $\left\{ f\sqb{k} \right\} _{k=0}^{K-1}$ is $\dft{f}{n} \DE \sum\nolimits_{k=0}^{K-1} {f\sqb{k}{e^{ - \jmath \frac{2\pi}{K} k n}}}$.
The Paley-Wiener class is denoted by $\mathsf{PW}_{\Omega}$, that is, $\BL{f} \cap \LL{2}$. The Dirac distribution is denoted $\delta\rob{t}$. The first order difference sequence is defined as $\bar{f}\sqb{k}=\rob{\Delta f}\sqb{k} = f\sqb{k+1} - f\sqb{k}$, and $\rob{\Delta^N f}\sqb{k} = \rob{\Delta\rob{\Delta^{N-1} f}}\sqb{k}$ is the $N^{\textrm{th}}$ order difference sequence. The mean squared error (MSE) between $\mathbf{x},\mathbf{y} \in \R^{K}$  is  $\mse{\mathbf{x}}{\mathbf{y}} \DE \frac{1}{K} \sum_{k=0}^{K-1}\left| x\sqb{k} - y\sqb{k} \right|^2$.

\ifCLASSOPTIONcaptionsoff
  \newpage
\fi

\end{document}